\documentclass[a4paper,UKenglish,cleveref, autoref, thm-restate]{lipics-v2021}

\usepackage[utf8]{inputenc}

\usepackage{textpos}
\usepackage{amsmath,amsthm,amssymb}
\usepackage{mathtools}
\usepackage{stmaryrd}
\usepackage{fancyvrb}
\usepackage[dvipsnames]{xcolor}
\usepackage{dashrule}
\usepackage{amsfonts}
\usepackage{amssymb}
\usepackage{bbold}
\usepackage{longtable}
\usepackage{mathpartir}
\usepackage{mathtools}

\usepackage{multicol}
\usepackage{xhfill}
\usepackage{proof}
\usepackage{cmll}
\usepackage[framemethod=tikz]{mdframed}
\usepackage{tikz}
\usepackage[colorinlistoftodos,prependcaption]{todonotes}
\presetkeys%
{todonotes}%
{inline,backgroundcolor=blue!20}{}
\usetikzlibrary{calc,shapes,arrows,positioning,automata}

\nolinenumbers

\newcommand{\tensor}{\ensuremath{\otimes}}

\newsavebox\MBox
\newcommand\Cline[2][red]{{\sbox\MBox{$#2$}%
  \rlap{\usebox\MBox}\color{#1}\rule[-1\dp\MBox-2pt]{\wd\MBox}{0.5pt}}}

\theoremstyle{definition}
\newtheorem{convention}[theorem]{Convention}

\usetikzlibrary{decorations.markings,arrows,decorations.pathmorphing}
\usetikzlibrary{shapes.misc,matrix,backgrounds,folding,positioning}
\usetikzlibrary{shapes.geometric}
\pgfdeclarelayer{edgelayer}
\pgfdeclarelayer{nodelayer}
\pgfsetlayers{background,edgelayer,nodelayer,main}
\tikzset{every path/.style={draw=black!80, line width=0.6pt}}
\tikzstyle{every picture}=[baseline=-0.25em]
\tikzstyle{none}=[inner sep=0mm]
\tikzstyle{zxnode}=[shape=circle, minimum width=.25cm, inner sep=0.5pt, font=\footnotesize, draw=black,thick]
\tikzstyle{gn}=[zxnode ,fill=green, draw=green!10!black]
\tikzstyle{rn}=[zxnode ,fill=red, draw=red!10!black]
\tikzstyle{H box}=[rectangle,fill=yellow, draw=yellow!10!black,thick,xscale=1,yscale=1,font=\footnotesize,inner sep=1.2pt,minimum width=0.15cm,minimum height=0.15cm]
\tikzstyle{arrow}=[decoration={markings,mark=at position 1 with
    {\arrow[scale=1.2,>=stealth]{>}}},postaction={decorate}]
\tikzstyle{glabel}=[rounded corners=0.2em,fill=green!30,inner sep=0.1em,font=\scriptsize, anchor=west, xshift=-0.3em, yshift=0,opacity=1]
\tikzstyle{rlabel}=[rounded corners=0.2em,fill=red!30,inner sep=0.1em,font=\scriptsize, anchor=west, xshift=-0.3em, yshift=0,opacity=1]

\pgfdeclarelayer{edgelayer}
\pgfdeclarelayer{nodelayer}
\pgfsetlayers{background,edgelayer,nodelayer,main}
\tikzstyle{every loop}=[]

\usetikzlibrary{circuits.ee.IEC}

\newcommand{\ground}
{
\begin{tikzpicture}[circuit ee IEC,yscale=0.9,xscale=0.8]
\draw[solid,arrows=-] (0,1ex) to (0,0) node[anchor=center,ground,rotate=-90,xshift=.66ex] {};
\end{tikzpicture}}

\newcommand{\sground}{\hspace*{-1pt}\scalebox{0.5}{\ground}}

\newcommand{\ket}[1]{\ensuremath{\left|  #1 \right\rangle}}
\newcommand{\bra}[1]{\ensuremath{\left\langle  #1 \right|}}
\newcommand{\ketbra}[2]{\ket{#1}\!\!\bra{#2}}
\newcommand{\braket}[2]{\left\langle\begin{array}{@{}c@{~}|@{~}c@{}} #1 & #2 \end{array}\right\rangle}
\newcommand{\tr}{\operatorname{Tr}}

\newcommand{\cat}[1]{\mathbf{#1}}
\newcommand{\interp}[1] {\left\llbracket #1 \right\rrbracket}
\newcommand{\ccpm}[1]{\operatorname{CPM}\left(#1\right)}

\newcommand{\rewrites}{\leadsto}

\newcommand{\cpm}{\ensuremath{\rewrites_{\sground}}}

\newcommand{\setirwer}{\rotatebox[origin=c]{180}{$\rewrites$}}

\newcommand{\downrewrites}{\rotatebox[origin=c]{-90}{$\rewrites$}}
\newcommand{\uprewrites}{\rotatebox[origin=c]{90}{$\rewrites$}}
\newcommand{\nerewrites}{\rotatebox[origin=c]{45}{$\rewrites$}}
\newcommand{\serewrites}{\rotatebox[origin=c]{-45}{$\rewrites$}}

\usetikzlibrary{cd}

\allowdisplaybreaks

\title{Geometry of Interaction for ZX{-}Diagrams}

\author{Kostia Chardonnet}{Université Paris-Saclay, CNRS, ENS Paris-Saclay, LMF, 91190, Gif-sur-Yvette, France. \\ Université de Paris, CNRS, IRIF, F-75006, Paris, France  \and \url{https://www.lri.fr/~chardonnet/}}{kostia@lri.fr}{}{}

\author{Benoît Valiron}{Université Paris-Saclay, CNRS, CentraleSupélec, ENS Paris-Saclay, LMF, 91190, Gif-sur-Yvette, France  \and \url{https://www.monoidal.net/}}{benoit.valiron@universite-paris-saclay.fr}{https://orcid.org/0000-0002-1008-5605}{}

\author{Renaud Vilmart}{Université Paris-Saclay, CNRS, ENS Paris-Saclay, Inria, LMF, 91190, Gif-sur-Yvette, France  \and \url{https://rvilmart.github.io/}}{vilmart@lsv.fr}{https://orcid.org/0000-0002-8828-4671}{}

\authorrunning{K. Chardonnet and B. Valiron and R. Vilmart} 

\Copyright{Kostia Chardonnet and Benoît Valiron and Renaud Vilmart} 



\relatedversion{}
\relatedversiondetails{Full Version}{https://hal.archives-ouvertes.fr/hal-03154573/}

\EventEditors{Filippo Bonchi and Simon J. Puglisi}
\EventNoEds{2}
\EventLongTitle{46th International Symposium on Mathematical Foundations of Computer Science (MFCS 2021)}
\EventShortTitle{MFCS 2021}
\EventAcronym{MFCS}
\EventYear{2021}
\EventDate{August 23--27, 2021}
\EventLocation{Tallinn, Estonia}
\EventLogo{}
\SeriesVolume{202}
\ArticleNo{16}

\begin{document}

\maketitle



\begin{abstract}
  ZX-Calculus is a versatile graphical language for quantum
  computation equipped with an equational theory.
  Getting inspiration from Geometry of Interaction, in this paper we
  propose a token-machine-based asynchronous model of both pure
  ZX-Calculus and its extension to mixed processes.
  We also show how to connect this new semantics to the usual standard
  interpretation of ZX-diagrams. This model allows us to have a new
  look at what ZX-diagrams compute, and give a more local, operational
  view of the semantics of ZX-diagrams.
\end{abstract}

\section{Introduction}

Quantum computing is a model of computation where data is stored on
the state of particles governed by the laws of quantum physics. The
theory is well established enough to have allowed the
design of quantum algorithms whose applications are gathering
interests from both public and private
actors~\cite{mckinsey,Qureca,quantumGoldRush}.

One of the fundamental properties of quantum objects is to have a
\emph{dual} interpretations.
In the first one, the quantum object is
understood as a \emph{particle}: a definite, localized
point in space, distinct from the other particles. Light can be for
instance regarded as a set of photons. In the other interpretation,
the object is understood as a \emph{wave}: it is ``spread-out'' in space,
possibly featuring interference. This is for instance the
interpretation of light as an electromagnetic wave.

The standard model of computation uses \emph{quantum bits} (qubits)
for storing information and \emph{quantum circuits}~\cite{nielsen} for
describing quantum operations with quantum gates, the quantum
version of Boolean gates.
Although the pervasive model for quantum computation, quantum
circuits' operational semantics is only given in an intuitive
manner. A quantum circuit is understood as some sequential, low-level
assembly language where quantum gates are opaque black-boxes.
In particular, quantum circuits do not natively feature any formal
operational semantics giving rise to abstract reasoning, equational
theory or well-founded rewrite system.

From a denotational perspective, quantum circuits are literal
description of tensors and applications of linear operators. These can
be described with the historical matrix interpretation~\cite{nielsen},
or with the more recent sum-over-paths semantics~\cite{amy-path-sum,
  chareton-qbricks} ---this can be regarded as a \emph{wave-style
  semantics}. In such a semantics, the state of all of the quantum
bits of the memory is mathematically represented as a vector in a
(finite dimensional) Hilbert space: the set of quantum bits is a
\emph{wave} flowing in the circuit, from the inputs to the output,
while the computation generated by the list of quantum gates is a
linear map from the Hilbert space of inputs to the Hilbert space of
outputs.

In recent years, an alternative model of quantum computation with
better formal properties than quantum circuits has emerged: the
ZX-Calculus~\cite{zxorigin}. Originally motivated by a categorical
interpretation of quantum theory, the ZX-Calculus is a graphical
language that represents linear maps as special kinds of graphs called
\emph{diagrams}. Unlike the quantum circuit framework, the ZX-Calculus
comes with a sound and complete~\cite{HNW,vilmart2019nearminimal}, well-defined equational
theory on a small set of canonical generators making it possible to
reason on quantum computation by means of local graph rewriting.

The canonical semantics of a ZX diagram consists in a linear
operator. This operator can be represented as a matrix or through the
more recent sum-over-path semantics~\cite{vilmart2020structure}. But
in both cases, these semantics give a purely functional,
\emph{wave-style} interpretation to the diagram.  Nonetheless, this
graphical language ---and its equational theory--- has been shown to
be amenable to many extensions and is being used in a wide spectrum of
applications ranging from quantum circuit
optimization~\cite{quantumcircuits1, quantumcircuits2},
verification~\cite{MSc.Hillebrand,duncan2014verifying,duncan2018verifying}
and representation such as MBQC patterns~\cite{duncan2010rewriting} or
error-correction~\cite{deBeaudrap2020zxlattice,dBDHP19}.

The standard models for both quantum circuits and ZX-Calculus is
therefore based on a wave-style interpretation.
An alternative operational interpretation of quantum circuit following
a \emph{particle-style} semantics has recently been investigated in
the literature~\cite{goisync}. In this model, quantum bits are
intuitively seen as \emph{tokens} flowing inside the wires of the
circuit. Formally, a quantum circuit is interpreted as a token-based
automata, based on Geometry of Interaction (GoI)~\cite{goi0, goi1,
  goi2, goi3}. Among its many instantiations, GoI
can be seen as a procedure to interpret a
proof-net~\cite{proofnets}---graphical representation of proofs of
linear logic~\cite{linearlogic}---as a token-based
automaton~\cite{danos1999reversible, asperti1995paths}. The flow of a
token inside a proof-net characterises an invariant of the proof---its
computational content. This framework is used in~\cite{goisync} to
formalize the notion of qubits-as-tokens flowing inside a higher-order
term representing a quantum computation---that is, computing a quantum
circuit. However, in this work, quantum gates are still regarded as
black-boxes, and tokens are purely classical objects \emph{requiring
  synchronicity}: to fire, a two-qubit gate needs its two arguments to
be ready.

As a summary, despite their ad-hoc construction, quantum circuits
can be seen from two perspectives: computation as a flow of particles
(i.e. tokens), and as a wave passing through the gates. On the other
hand, although ZX-Calculus is a well-founded language, it still misses such
a particle-style perspective.

\emph{In this paper, we aim at giving a novel insight on the
  computational content of a ZX term in an asynchronous way,
  emphasizing the graph-like behavior of a ZX-diagram.}

Following the idea of using a token machine to exhibit the
computational content of a proof-net or a quantum circuit, we present
in this paper a token machine for the ZX-Calculus. To exemplify the
versatility of the approach, we show how to extend it to mixed
processes~\cite{coecke2012environment,carette2019completeness}. To
assess the validity of the semantics, we show how it links to the
standard interpretation of ZX-diagrams. While the standard
interpretation of ZX-diagrams proceeds with diagram decomposition as tensors and products of matrices, the tokens flowing inside the diagram really exploits the connectivity of the diagram.

\noindent
\begin{minipage}{.7\textwidth}
~~~This ability illustrates one fundamental difference between our
approach and the one in~\cite{goisync}. The latter follows a
\emph{classical control} approach: if qubits can be in superposition,
each qubit inhabits a token sitting in \emph{one single} position in
the circuit. For instance, on the circuit on the right, the state of
the two tokens
\end{minipage}
\hfill
\raisebox{-2ex}{\includegraphics[scale=.6]{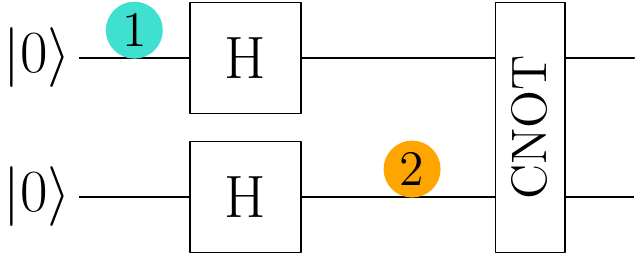}}

\noindent
$|{\color{Turquoise}\bullet}{\color{orange}\bullet}\rangle$ is $\frac{\sqrt2}2(\ket{00}+\ket{10})$. Although the two tokens can be
regarded as being in superposition, their \emph{position} is not.
In our system, tokens and positions can be superposed.
The second fundamental difference lies in the \emph{asynchronicity} of
our token-machine. Unlike~\cite{goisync}, we rely on the canonical
generators of ZX-diagrams: tokens can travel through these nodes in an
asynchronous manner. For instance, in the above circuit the orange
token has to wait for the blue token before crossing the CNOT
gate. As illustrated in Table~\ref{tab:async-rules}, in our system one
token can interact with multi-wire nodes.
Finally, as formalized in Theorem~\ref{thm:arbitrary-wire-init}, a
third difference is that compared to~\cite{goisync}, the token-machine
we present is \emph{non-oriented}: in the circuit above, tokens have
to start on the left and flow towards the right of the circuit whereas
our system is agnostic on where tokens initially ``start''.

\medskip
\noindent
\textbf{Plan of the paper.}
The paper is organized as follows : in Section~\ref{sec:ZX} we present
the ZX-Calculus and its standard interpretation into
$\cat{Qubit}$, and its axiomatization.

In Section~\ref{sec:token-machine} we present the actual asynchronous token machine
and its semantics and show that it is sound and complete with regard
to the standard interpretation of ZX-diagrams. Finally, in Section~\ref{sec:mixed-processes} we present an extension of the
ZX-Calculus to mixed processes and adapt the token machine to take
this extension into account.

Proofs are in the appendix.

\section{The ZX-Calculus}
\label{sec:ZX}
The ZX-Calculus is a powerful graphical language for reasoning about
quantum computation introduced by Bob Coecke and Ross
Duncan~\cite{zxorigin}. A term in this language is a graph ---called a
\emph{string diagram}--- built from a core set of primitives. In the
standard interpretation of ZX-Calculus, a string diagram is
interpreted as a matrix. The language is equipped with an equational
theory preserving the standard interpretation.

\subsection{Pure Operators}

\begin{minipage}{0.4\columnwidth}
  The so-called \emph{pure} ZX-diagrams are generated from a set of
  primitives, given on the right: the Identity, Swap, Cup, Cap, Green-spider and
H-gate:
\end{minipage}
\hfill
$\displaystyle
  \left\lbrace 
\begin{tikzpicture}
	\begin{pgfonlayer}{nodelayer}
		\node [style=none] (0) at (0.125, 0.25) {};
		\node [style=none] (1) at (0.125, -0.25) {};
		\node [style=none, font={\scriptsize}, color=gray] (2) at (0, 0) {$e_0$};
	\end{pgfonlayer}
	\begin{pgfonlayer}{edgelayer}
		\draw (0.center) to (1.center);
	\end{pgfonlayer}
\end{tikzpicture}
\colorbox{pink}{missing file : id-nw}}}
, 
\InputIfFileExists{./figures/swap-nw.tikz}{}{{\color{red}\colorbox{pink}{missing file : swap-nw}}}
, 
\InputIfFileExists{./figures/cup-nw.tikz}{}{{\color{red}\colorbox{pink}{missing file : cup-nw}}}
,
    
\InputIfFileExists{./figures/cap-nw.tikz}{}{{\color{red}\colorbox{pink}{missing file : cap-nw}}}
, 
\InputIfFileExists{./figures/gn-alpha-nw.tikz}{}{{\color{red}\colorbox{pink}{missing file : gn-alpha-nw}}}
,
    
\InputIfFileExists{./figures/H-nw.tikz}{}{{\color{red}\colorbox{pink}{missing file : H-nw}}}
\right\rbrace_{\substack{n,m\in\mathbb
      N\\\alpha\in\mathbb R\\e_i,e'_i\in \mathcal E}}
$\\
We shall be using the following labeling convention: wires (edges) are labeled
with $e_i$, taken from an infinite set of labels $\mathcal E$. We take
for granted that distinct wires have distinct labels. The real number
$\alpha$ attached to the green spiders is called the
\emph{angle}. ZX-diagrams are read top-to-bottom: dangling top edges
are the \emph{input edges} and dangling edges at the bottom are
\emph{output edges}. For instance, Swap has 2 input and 2 output
edges, while Cup has 2 input edges and no output edges.
We write $\mathcal E(D)$ for the set of edge labels in the diagram
$D$, and $\mathcal I(D)$ (resp.~$\mathcal O(D)$) for the list of input
edges (resp.~output edges) of $D$. We denote $::$ the concatenation of lists.

ZX-primitives can be composed either sequentially or in parallel :

\begin{minipage}{0.47\textwidth}
 $$D_2\circ D_1:=
\InputIfFileExists{./figures/compo.tikz}{}{{\color{red}\colorbox{pink}{missing file : compo}}}
$$
\end{minipage}
\begin{minipage}{0.47\textwidth}
$$D_1\otimes D_2:=
\InputIfFileExists{./figures/tensor.tikz}{}{{\color{red}\colorbox{pink}{missing file : tensor}}}
$$
\end{minipage}

We write $\cat{ZX}$ for the set of all ZX-diagrams.
Notice that when composing diagrams with $(\_\circ\_)$, we ``join''
the outputs of the top diagram with the inputs of the bottom
diagram. This requires that the two sets of edges have the same
cardinality. The junction is then made by relabeling the input edges
of the bottom diagram by the output labels of the top diagram.

\noindent\begin{minipage}{0.7\columnwidth}
  \begin{convention}
  We define a second spider, red this time, by composition of Green-spiders and H-gates, as shown on the right.
\end{convention}
\end{minipage}
\hfill 
\InputIfFileExists{./figures/colour-change.tikz}{}{{\color{red}\colorbox{pink}{missing file : colour-change}}}

\begin{convention}
  We write $\sigma$ for a permutation of wires, i.e any diagram
  generated by
  $\left\lbrace~
\begin{tikzpicture}
	\begin{pgfonlayer}{nodelayer}
		\node [style=none] (0) at (0, 0.25) {};
		\node [style=none] (1) at (0, -0.25) {};
	\end{pgfonlayer}
	\begin{pgfonlayer}{edgelayer}
		\draw (0.center) to (1.center);
	\end{pgfonlayer}
\end{tikzpicture}
\colorbox{pink}{missing file : id}}}
~,
\InputIfFileExists{./figures/swap.tikz}{}{{\color{red}\colorbox{pink}{missing file : swap}}}
\right\rbrace$ with
  sequential and parallel composition. We write the Cap as $\eta$ and
  the Cup as $\epsilon$. We write $Z_k^n(\alpha)$ (resp, $X_k^n$) for
  the green-node (resp, red-node) of $n$ inputs, $k$ outputs and
  parameter $\alpha$ and $H$ for the H-gate. In
  the remainder of the paper we omit the edge labels when not necessary . Finally, by
  abuse of notation a green or red node with no explicit parameter
  holds the angle
  $0$: $\quad
\InputIfFileExists{./figures/gn-0-def.tikz}{}{{\color{red}\colorbox{pink}{missing file : gn-0-def}}}
\quad\text{ and }\quad
\InputIfFileExists{./figures/rn-0-def.tikz}{}{{\color{red}\colorbox{pink}{missing file : rn-0-def}}}
$.
\end{convention}

\subsection{Standard Interpretation}
\label{sec:std-interp}
We understand ZX-diagrams as linear operators through the
\emph{standard interpretation}. Informally, wires are interpreted with
the two-dimensional Hilbert space, with orthonormal basis written as
$\{\ket0, \ket1\}$, in Dirac notation~\cite{nielsen}. Vectors of the
form $\ket{.}$ (called ``\emph{kets}'') are considered as columns
vector, and therefore $\ket0 =
\left(\begin{smallmatrix}1\\0\end{smallmatrix}\right)$, $\ket1 =
\left(\begin{smallmatrix}0\\1\end{smallmatrix}\right)$, and
$\alpha\ket0+\beta\ket1 =
\left(\begin{smallmatrix}\alpha\\\beta\end{smallmatrix}\right)$.
Horizontal juxtaposition of wires is interpreted with the
\emph{Kronecker}, or \emph{tensor} product. The tensor product of
spaces $\mathcal V$ and $\mathcal W$ whose bases are respectively
$\{v_i\}_i$ and $\{w_j\}_j$ is the vector space of basis $\{v_i\otimes
w_j\}_{i,j}$, where $v_i\otimes w_j$ is a formal object consisting of
a pair of $v_i$ and $w_j$. We denote $\ket{x}\otimes\ket{y}$ as
$\ket{xy}$. In the interpretation of spiders, we use the notation
$\ket{0^m}$ to represent an $m$-fold tensor of $\ket{0}$.
As a shortcut notation, we write $\ket{\phi}$ for column
vectors consisting of a linear combinations of kets.
Shortcut notations are also used for two very useful states:
$\ket+{}:=\frac{\ket0{}+\ket1{}}{\sqrt2}$ and
$\ket-{}:=\frac{\ket0{}-\ket1{}}{\sqrt2}$.
Dirac also introduced the notation ``\emph{bra}'' $\bra{x}$, standing
for a row vector. So for instance, $\alpha\bra{0} + \beta\bra{1}$ is
$(\begin{smallmatrix}\alpha&\beta\end{smallmatrix})$. If $\ket{\phi} =
\alpha\ket0+\beta\ket1$, we then write $\bra\phi$ for the vector
$\overline{\alpha}\bra0+\overline{\beta}\bra1$ (with $\overline{(.)}$
the complex conjugation). The notation for tensors of bras is similar
to the one for kets. For instance, $\bra{x}\tensor\bra{y} = \bra{xy}$.
Using this notation, the scalar product is transparently the product
of a row and a column vector: $\braket\phi\psi$, and matrices can be
written as sums of elements of the form $\ketbra\phi\psi$. For
instance, the identity on $\mathbb{C}^2$ is
$\left(\begin{smallmatrix}1&0\\0&1\end{smallmatrix}\right) = {}$
$\left(\begin{smallmatrix}1&0\\0&0\end{smallmatrix}\right) +
\left(\begin{smallmatrix}0&0\\0&1\end{smallmatrix}\right) = {}$
$\left(\begin{smallmatrix}1\\0\end{smallmatrix}\right)\left(\begin{smallmatrix}1&0\end{smallmatrix}\right)
+
\left(\begin{smallmatrix}0\\1\end{smallmatrix}\right)\left(\begin{smallmatrix}0&1\end{smallmatrix}\right)
= {}$
$\ketbra00 + \ketbra11$.
For more information on how Hilbert spaces, tensors, compositions and
bras and kets work, we invite the reader to consult e.g.~\cite{nielsen}.

In the {standard interpretation}~\cite{zxorigin}, a diagram $D$ is mapped to a map between
finite dimensional Hilbert spaces of dimensions some powers of $2$:
$\interp{D}\in\cat{Qubit}:=\{\mathbb{C}^{2^n}\to\mathbb{C}^{2^m}\mid n,m\in\mathbb N\}$.

If $D$ has $n$ inputs and $m$ outputs, its interpretation is a map
$\interp{D}:\mathbb{C}^{2^n}\to\mathbb{C}^{2^m}$ (by abuse of
notation we shall use the notation $\interp{D}:n\to m$).
It is defined inductively as follows.

~\hfill
$\interp{
\InputIfFileExists{./figures/compo.tikz}{}{{\color{red}\colorbox{pink}{missing file : compo}}}
}=\interp{
\InputIfFileExists{./figures/D2.tikz}{}{{\color{red}\colorbox{pink}{missing file : D2}}}
}\circ\interp{
\InputIfFileExists{./figures/D1.tikz}{}{{\color{red}\colorbox{pink}{missing file : D1}}}
}$
\hfill
$\interp{
\InputIfFileExists{./figures/tensor.tikz}{}{{\color{red}\colorbox{pink}{missing file : tensor}}}
}=\interp{
\InputIfFileExists{./figures/D1.tikz}{}{{\color{red}\colorbox{pink}{missing file : D1}}}
}\otimes\interp{
\InputIfFileExists{./figures/D2.tikz}{}{{\color{red}\colorbox{pink}{missing file : D2}}}
}$
\hfill~

~\hfill $\interp{~
\begin{tikzpicture}
	\begin{pgfonlayer}{nodelayer}
		\node [style=none] (0) at (0, 0.25) {};
		\node [style=none] (1) at (0, -0.25) {};
	\end{pgfonlayer}
	\begin{pgfonlayer}{edgelayer}
		\draw (0.center) to (1.center);
	\end{pgfonlayer}
\end{tikzpicture}
\colorbox{pink}{missing file : id}}}
~}=id_{\mathbb C^2}=\ketbra00+\ketbra11$
\hfill
$\displaystyle\interp{~
\InputIfFileExists{./figures/swap.tikz}{}{{\color{red}\colorbox{pink}{missing file : swap}}}
~}=\sum_{i,j\in\{0,1\}}\ketbra{ji}{ij}$
\hfill~

~\hfill
$\interp{~
\begin{tikzpicture}
	\begin{pgfonlayer}{nodelayer}
		\node [style=none] (0) at (-0.25, -0.125) {};
		\node [style=none] (1) at (0.25, -0.125) {};
		\node [style=none] (2) at (0, 0.125) {};
	\end{pgfonlayer}
	\begin{pgfonlayer}{edgelayer}
		\draw [bend left=90, looseness=1.75] (0.center) to (1.center);
	\end{pgfonlayer}
\end{tikzpicture}
\colorbox{pink}{missing file : cap}}}
~}=\interp{~
\begin{tikzpicture}
	\begin{pgfonlayer}{nodelayer}
		\node [style=none] (0) at (-0.25, 0.125) {};
		\node [style=none] (1) at (0.25, 0.125) {};
		\node [style=none] (2) at (0, -0.125) {};
	\end{pgfonlayer}
	\begin{pgfonlayer}{edgelayer}
		\draw [bend right=90, looseness=1.75] (0.center) to (1.center);
	\end{pgfonlayer}
\end{tikzpicture}
\colorbox{pink}{missing file : cup}}}
~}^\dagger
=\ket{00}{}+\ket{11}{}$ \hfill
$\interp{~
\begin{tikzpicture}
	\begin{pgfonlayer}{nodelayer}
		\node [style=none] (0) at (0, 0.25) {};
		\node [style=none] (1) at (0, -0.25) {};
		\node [style={{H box}}] (2) at (0, 0) {};
	\end{pgfonlayer}
	\begin{pgfonlayer}{edgelayer}
		\draw (0.center) to (1.center);
	\end{pgfonlayer}
\end{tikzpicture}
\colorbox{pink}{missing file : H}}}
~}=\ketbra+0+\ketbra-1$ \hfill~

~\hfill
$\interp{
\InputIfFileExists{./figures/gn-alpha.tikz}{}{{\color{red}\colorbox{pink}{missing file : gn-alpha}}}
}=\ketbra{0^m}{0^n}+e^{i\alpha}\ketbra{1^m}{1^n}$
\hfill
$\interp{
\InputIfFileExists{./figures/rn-alpha.tikz}{}{{\color{red}\colorbox{pink}{missing file : rn-alpha}}}
}=\ketbra{+^m}{+^n}+e^{i\alpha}\ketbra{-^m}{-^n}$
\hfill~

\subsection{Properties and structure}

In this section, we list several definitions and known results that we
shall be using in the remainder of the paper. See e.g.~\cite{vilmart-thesis} for
more information.

\smallskip
\noindent
{\bf Universality}.~
ZX-diagrams are \emph{universal} in the sense
that for any linear map $f:n\to m$, there exists a diagram $D$ of $\cat{ZX}$ such
that $\interp{D}=f$.

The price to pay for universality is that different diagrams can
possibly represent the same quantum operator. There exists however a
way to deal with this problem: an equational theory. Several
equational theories have been designed for various fragments of the
language \cite{backens2014stabilizer,JPV,HNW,JPV-universal,jeandel2019generic,vilmart2019nearminimal}.

\smallskip
\noindent
{\bf Core axiomatization.}~
Despite this variety, any ZX axiomatization builds upon the core set
of equations provided in Figure~\ref{fig:connectivity-rules}, meaning
that edges really behave as wires that can be bent, tangled and
untangled. They also enforce the irrelevance on the ordering of
inputs and outputs for spiders. Most importantly, these rules
preserve the standard interpretation given in
Section~\ref{sec:std-interp}. We will use these rules ---sometimes
referred to as ``\emph{only connectivity matters}''---, and the fact
that they preserve the semantics extensively in the proofs of
the results of the paper.

\begin{figure}[tb]
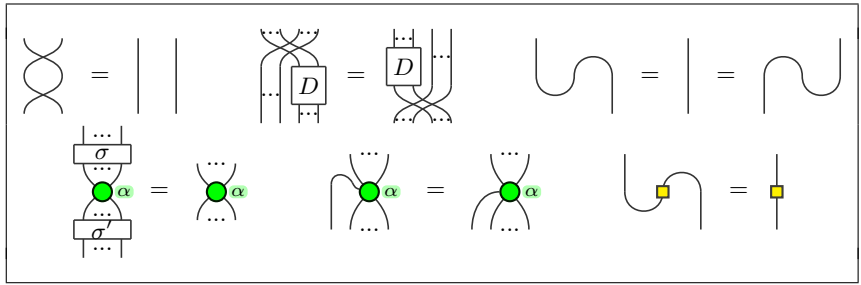

    \centering
    \begin{tabular}{|c|}
       \hline
        \\[-1ex]
      
\InputIfFileExists{./figures/swap-involution.tikz}{}{{\color{red}\colorbox{pink}{missing file : swap-involution}}}
$\qquad$
      
\InputIfFileExists{./figures/naturality-swap.tikz}{}{{\color{red}\colorbox{pink}{missing file : naturality-swap}}}
$\qquad$
      
\InputIfFileExists{./figures/snake.tikz}{}{{\color{red}\colorbox{pink}{missing file : snake}}}
\\[1ex]
      
\InputIfFileExists{./figures/gn-permutation-of-wires.tikz}{}{{\color{red}\colorbox{pink}{missing file : gn-permutation-of-wires}}}
$\qquad$
      
\InputIfFileExists{./figures/gn-input-output.tikz}{}{{\color{red}\colorbox{pink}{missing file : gn-input-output}}}
$\qquad$
      
\InputIfFileExists{./figures/H-input-output.tikz}{}{{\color{red}\colorbox{pink}{missing file : H-input-output}}}
\\[-1ex]
        \\
       \hline
  \end{tabular}
  \caption[Connectivity rules]{Connectivity rules. $D$ represents any
    ZX-diagram, and $\sigma,\sigma'$ any permutation of wires.}
    \label{fig:connectivity-rules}
\end{figure}

\smallskip
\noindent
{\bf Completeness.}~
The ability to transform a diagram $D_1$ into a
diagram $D_2$ using the rules of some axiomatization
$\operatorname{zx}$ (e.g. the core one presented in
Figure~\ref{fig:connectivity-rules}) is denoted
$\operatorname{zx}\vdash D_1=D_2$.

The axiomatization is said to be \emph{complete} whenever any two
diagrams representing the same operator can be turned into one another
using this axiomatization. Formally:\\
\centerline{$\displaystyle\interp{D_1}=\interp{D_2}\iff\operatorname{zx}\vdash D_1=D_2$}

It is common in quantum computing to work with restrictions of quantum
mechanics. Such restrictions translate to restrictions to particular
sets of diagrams -- e.g.~the $\frac\pi4$-fragment which consists of
all ZX-diagrams where the angles are multiples of $\frac\pi4$. There
exists axiomatizations that were proven to be complete for the
corresponding fragment (all the aforementioned references tackle the problem of completeness).

The developments of this paper are given for the ZX-Calculus in its
most general form, but everything in the following also
works for fragments of the language.

\smallskip
\noindent
{\bf Input and output wires.}~
An important result which will be used in the rest of the paper is the
following:

\begin{theorem}[Choi-Jamio\l{}kowski]\label{th:iso-wire}
  There are isomorphisms between $\{D\in\cat{ZX}\mid D:n\to m\}$ and
  $\{D\in\cat{ZX}\mid D:n-k\to k+m\}$ (when $k\leq n$).\qed
\end{theorem}
To see how this can be true, simply add cups or caps to turn input
edges to output edges (or vice versa), and use the fact that we work
modulo the rules of Figure~\ref{fig:connectivity-rules}.

When $k=n$, this isomorphism is referred to as the \emph{map/state
  duality}. A related but more obvious isomorphism between ZX-diagrams
is obtained by permutation of input wires (resp.~output wires).

\subsection{Notions of Graph Theory in ZX}

Theorem~\ref{th:iso-wire} is essential: it allows us to transpose
notions of graphs into ZX-Calculus. It is for instance possible to
define a notion of connectivity.

\begin{definition}[Connected Components]
  Let $D$ be a non-empty $\cat{ZX}$-diagram. Consider all of the
  possible decompositions
  with $D_1,...,D_k\in\cat{ZX}$ and $\sigma,\sigma'$ permutations of
  wires:\\
  $ D = 
\InputIfFileExists{./figures/connected-components.tikz}{}{{\color{red}\colorbox{pink}{missing file : connected-components}}}
 $\hfill
  \begin{minipage}{0.8\columnwidth}
  The largest such $k$ is called the number of connected components of
  $D$. It induces a decomposition. The induced $D_1,...,D_n$ are called the connected components
  of $D$. If $D$ has only one connected component, we say that $D$ is
  connected.
  \end{minipage}
\end{definition}

We can also consider the notions of paths, distance and cycles of
usual multi-graphs. We denote $\operatorname{Paths}(e_0,e_n)$ the set
of paths from edge $e_0$ to $e_n$. We denote $\operatorname{Paths}(D)$
(resp. $\operatorname{Cycles}(D)$) the set of paths (resp. cycles) of
diagram $D$. For a path $p$, we denote $|p|$ its length. We denote
$d(e_0,e_n)$ the distance i.e. the length of the shortest path between
$e_0$ and $e_n$.

\section{A Token Machine for ZX-diagrams}
\label{sec:token-machine}

Inspired by the Geometry of Interaction~\cite{goi0, goi1, goi2, goi3} and the associated notion of token
machine~\cite{danos1999reversible,asperti1995paths} for proof
nets~\cite{proofnets}, we define here a first token machine on pure
ZX-diagrams. The token consists of an edge of the diagram, a direction
(either going up, noted $\uparrow$, or down, noted $\downarrow$) and a
bit (state). The idea is that, starting from an input edge the token
will traverse the graph and duplicate itself when encountering an
n-ary node (such as the green and red) into each of the input / output
edges of the node. Notice that it is not the case for token machines
for proof-nets where the token never duplicates itself.  This
duplication is necessary to make sure we capture the whole linear map
encoded by the ZX-diagram. Due to this duplication, two tokens might
collide together when they are on the same edge and going in different
directions. The result of such a collision will depend on the states held by both tokens.
For a cup, cap or identity diagram, the
token will simply traverse it. As for the Hadamard node the token will
traverse it and become a superposition of two tokens with opposite
states.
Therefore, as tokens move through a diagram, some may be added,
multiplied together, or annihilated.

\begin{definition}[Tokens and Token States]
  Let $D$ be a ZX-diagram.  A \emph{token} in $D$ is a triplet
  $(e,d,b)\in \mathcal
  E(D)\times\{\downarrow,\uparrow\}\times\{0,1\}$. We shall omit the
  commas and simply write $(e~d~b)$.
  The set of tokens on $D$ is written $\cat{tk}(D)$.  A \emph{token
    state} $s$ is then a multivariate polynomial over $\mathbb C$,
  evaluated in $\cat{tk}(D)$.
  We define $\cat{tkS}(D) := \mathbb C[\cat{tk}(D)]$ the algebra of
  multivariate polynomials over $\cat{tk}(D)$.

  In the token state $t = \sum_i\alpha_i\,t_{1,i}\cdots{}t_{n_i,i}$,
  where the $t_{k,i}$'s are tokens, the components
  $\alpha_i\,t_{1,i}\cdots{}t_{n_i,i}$ are called the \emph{terms} of
  $t$.
\end{definition}

A monomial $(e_1\,d_1,b_1)\cdots(e_n\,d_n,b_n)$ encodes the state of
$n$ tokens in the process of flowing in the diagram $D$. A token state
is understood as a \emph{superposition} ---a linear combination--- of
multi-tokens flowing in the diagram.

\begin{convention}\label{conv:poly}
  In token states, the sum ($+$) stands for the superposition while the
  product stands for additional tokens within a given diagram. We follow the
  usual convention of algebras of polynomials: for instance, if $t_i$
  stands for some token $(e_i~d_i~b_i)$, then
  $
  (t_1+t_2)t_3 = (t_1t_2)+(t_1t_3),
  $
  that is, the superposition of $t_1$,$t_2$ flowing in $D$ and
  $t_1$,$t_3$ flowing in $D$. Similarly, we consider token states
  modulo commutativity of sum and product, so that for instance the
  monomial $t_1t_2$ is the same as $t_2t_1$. Notice that $0$ is an
  absorbing element for the product ($0\times t = 0$) and that $1$ is
  a neutral element for the same operation ($1\times t = t$).
\end{convention}

\subsection{Diffusion and Collision Rules}

The tokens in a ZX-diagram $D$ are meant to move inside $D$. The set
of rules presented in this section describes an \emph{asynchronous}
evolution, meaning that given a token state, we will rewrite only one
token at a time. The synchronous setting is discussed in
Section~\ref{sec:future-work}.

\begin{definition}[Asynchronous Evolution]\label{def:async}
  Token states on a diagram $D$ are equipped with two transition
  systems:
  \begin{itemize}
  \item a \emph{collision system} ($\rewrites_c$), whose effect is to annihilate tokens;
  \item a \emph{diffusion sub-system} ($\rewrites_d$), defining the
    flow of tokens within $D$.
  \end{itemize}
  The two systems are defined as follows. With $X\in \{d, c\}$
  and $1\leq j\leq n_i$,
  if $t_{i,j}$ are tokens in $\cat{tk}(D)$, then
  using Convention~\ref{conv:poly},\\
  \centerline{$\displaystyle
    \sum_i\alpha_it_{i,1}\cdots t_{i,j}\cdots t_{i,n_i}
    \rewrites_X
    \sum_i\alpha_it_{i,1}\cdots\left(\sum_k\beta_k t'_k\right)\cdots t_{i,n_i}
  $}
  provided that $t_{i,j}\rewrites_X \sum_k\beta_k t'_k$ according to the
  rules of Table~\ref{tab:async-rules}. In the table, each rule
  corresponds to the interaction with the primitive diagram
  constructor on the left-hand-side. Variables $x$ and $b$ span
  $\{0,1\}$, and $\neg$ stands for the negation.  In the green-spider
  rules, $e^{i\alpha{}x}$ stands for the the complex number
  $\cos(\alpha{}x)+i\sin(\alpha{}x)$ and not an edge label.

  Finally, as it is customary for rewrite systems, if $(\to)$ is a
  step in a transition system, $(\to^*)$ stands for the reflexive,
  transitive closure of $(\to)$.
\end{definition}

\begin{table}[t]
\scalebox{.9}{\begin{minipage}{1.08\textwidth}
\begin{align*}
  
\InputIfFileExists{./figures/id-nw.tikz}{}{{\color{red}\colorbox{pink}{missing file : id-nw}}}

  &\quad
  & (e_0\downarrow x)(e_0\uparrow x)
  &{}\rewrites_c 1\quad
    (e_0\downarrow x)(e_0\uparrow \neg x)
    \rewrites_c 0
    \tag{Positive/Negative Collision}
  \\
  
\InputIfFileExists{./figures/cup-nw.tikz}{}{{\color{red}\colorbox{pink}{missing file : cup-nw}}}

  && (e_b\downarrow x)
  &{}\rewrites_d (e_{\neg b}\uparrow x)
    \tag{$
\InputIfFileExists{./figures/cup.tikz}{}{{\color{red}\colorbox{pink}{missing file : cup}}}
$-diffusion}
  \\
  
\InputIfFileExists{./figures/cap-nw.tikz}{}{{\color{red}\colorbox{pink}{missing file : cap-nw}}}

  && (e_b\uparrow x)
  &{}\rewrites_d (e_{\neg b}\downarrow x)
    \tag{
\InputIfFileExists{./figures/cap.tikz}{}{{\color{red}\colorbox{pink}{missing file : cap}}}
-diffusion}
  \\
  && (e_k\downarrow x)
  &{}\rewrites_d e^{i\alpha x}\prod_{i\neq k} (e_i\uparrow x) \prod_j
    (e'_j\downarrow x)
  \\[-1ex]
  && (e'_k\uparrow x)
  &{}\rewrites_d e^{i\alpha x}\prod_{j\neq k} (e'_j\downarrow x) \prod_i
    (e_i\uparrow x)
  \\
  &&
     (e_0\downarrow x)
  &{}\rewrites_d (-1)^x\frac1{\sqrt2}(e_1\downarrow x) +
    \frac1{\sqrt2}(e_1\downarrow \neg x)
  \\
  &&
     (e_1\uparrow x)
  &{}\rewrites_d (-1)^x\frac1{\sqrt2}(e_0\uparrow x) +
    \frac1{\sqrt2}(e_0\uparrow \neg x)\hspace{10ex}~
\end{align*}%
\begin{textblock}{1}(1,-2.1)
\InputIfFileExists{./figures/gn-alpha-nw.tikz}{}{{\color{red}\colorbox{pink}{missing file : gn-alpha-nw}}}
\end{textblock}\begin{textblock}{1}(1,-1)
\InputIfFileExists{./figures/H-nw.tikz}{}{{\color{red}\colorbox{pink}{missing file : H-nw}}}
\end{textblock}%
\begin{textblock}{2}(9.8,-.8)(
\InputIfFileExists{./figures/H.tikz}{}{{\color{red}\colorbox{pink}{missing file : H}}}
-Diffusion)\end{textblock}%
\begin{textblock}{2}(9.5,-1.8)(
\InputIfFileExists{./figures/gn.tikz}{}{{\color{red}\colorbox{pink}{missing file : gn}}}
-Diffusion)\end{textblock}
\end{minipage}}
\caption{Asynchronous token-state evolution, for all $x,b\in\{0,1\}$}
\label{tab:async-rules}
\end{table}

We aim at a transition system marrying both collision and diffusion
steps. However, for consistency of the system, the order in which we
apply them is important as illustrated by the following example.

\begin{example}
  \label{ex:problem}
  Consider the equality given by the ZX equational theories:
  
\InputIfFileExists{./figures/gn-special.tikz}{}{{\color{red}\colorbox{pink}{missing file : gn-special}}}
.

  \noindent
  If we drop a token with bit $0$ at the top, we hence expect to get a
  single token with bit $0$ at the bottom. We underline the token that
  is being rewritten at each step. This is what we get when giving
  the priority to collisions:
  $$\begin{array}{c@{\quad ::\quad}l} 
\InputIfFileExists{./figures/gn-special-nw.tikz}{}{{\color{red}\colorbox{pink}{missing file : gn-special-nw}}}
\quad
      &\quad (a\downarrow 0)\rewrites_d \underline{(b\downarrow
      0)}(c\downarrow 0) \rewrites_d (d\downarrow
      0)\underline{(c\uparrow 0)(c\downarrow 0)} \rewrites_c
      (d\downarrow 0) \end{array}$$ Notice that the collision
      $(c\uparrow 0)(c\downarrow 0)$ rewrites to $1$, and therefore
      the product $(d\downarrow 0)\times 1 = (d\downarrow 0)$.
    If however we decide to ignore the priority of collisions, we may end
    up with a non-terminating run, unable to converge to $(d\downarrow 0)$:
    $$ (a\downarrow 0)\rewrites_d \underline{(b\downarrow 0)}(c\downarrow 0)
          \rewrites_d (d\downarrow 0)\underline{(c\uparrow 0)}(c\downarrow 0)
          \rewrites_d (d\downarrow 0)(a\uparrow 0)(b\downarrow 0)(c\downarrow 0)
          \rewrites_d \dots
    $$
\end{example}

We therefore set a rewriting strategy as follows.

\begin{definition}[Collision-Free]
A token state $s$ of $\cat{tkS}(D)$ is called \emph{collision-free} if
for all $s'\in\cat{tkS}(D)$, we have $s\not\rewrites_c s'$.
\end{definition}

\begin{definition}[Token Machine Rewriting System]
  \label{def:tmrw}
  We define a transition system $\rewrites$ as \emph{exactly} one
  $\rewrites_d$ rule followed by all possible $\rewrites_c$ rules. In
  other words,
  $t \rewrites u$ if and only if there exists $t'$ such that $t\rewrites_d t'
  \rewrites_c^* u$ and $u$ is collision-free.
\end{definition}

In~\cite{goisync}, a token arriving at an input of a gate is
blocked until all the inputs of the gates are populated by a 
token, at which point all the tokens go through at once
(while obviously changing the state). The control is purely
classical: it is causal. In our approach, the state of the
system is global and there is no explicit notion of qubit. 
Instead, tokens collect the operation that is to be applied to the input qubits.

\subsection{Strong Normalization and Confluence}

The token machine Rewrite System of Definition~\ref{def:tmrw} ensures
that the collisions that can happen always happen. The system does not
a priori forbid two tokens on the same edge, provided that they have
the same direction. However this is something we want to avoid as
there is no good intuition behind it: We want to link the token
machine to the standard interpretation, which is not possible if two
tokens can appear on the same edge.

In this section we show that, under a notion of well-formedness
characterizing token uniqueness on each edge, the
Token State Rewrite System ($\rewrites$) is strongly normalizing and
confluent.

\begin{definition}[Polarity of a Term in a Path]
  Let $D$ be a ZX-diagram, and $p\in\operatorname{Paths}(D)$ be a path in
  $D$. Let $t=(e,d,x)\in\cat{tk}(D)$. Then:\\
 \phantom{.}\hfill$\displaystyle P(p,t)=\begin{cases}
    1&\text{if $e\in p$ and $e$ is $d$-oriented}\\
    -1&\text{if $e\in p$ and $e$ is $\neg d$-oriented}\\
    0&\text{if $e\notin p$}
  \end{cases}$\hfill~\\
  We extend the definition to subterms $\alpha\,t_1...t_m$ of a
  token-state $s$:\\
  \phantom{.}\hfill$\displaystyle P(p,0)=P(p,1)=0,
  \qquad
  P(p,\alpha\, t_1...t_m)=P(p,t_1)+...+P(p,t_m).$\hfill~\\
  In the following, we shall simply refer to such subterms as ``\emph{terms of $s$}''.
\end{definition}

\noindent
\begin{minipage}{.7\textwidth}
  \begin{example}
    In the (piece of) diagram presented on the right, the blue
    directed line $p = (e_0,e_1,e_2,e_3,e_4)$ is a path. The
    orientation of the edges in the path is represented by the arrow
    heads, and $e_3$ for instance is $\downarrow$-oriented in $p$
    which implies that we have $P(p,(e_3\uparrow x)) = -1$.
  \end{example}
\end{minipage}
\hfill

\InputIfFileExists{./figures/example-path-polarity.tikz}{}{{\color{red}\colorbox{pink}{missing file : example-path-polarity}}}

\begin{definition}[Well-formedness]
  Let $D$ be a ZX-diagram, and $s\in\cat{tkS}(D)$ a token state on
  $D$. We say that $s$ is well-formed if for every term $t$ in $s$
  and every path $p\in\operatorname{Paths}(D)$ we have $P(p,t)\in\{-1,0,1\}$.
\end{definition}

\begin{proposition}[Invariance of Well-Formedness]
  \label{prop:well-formedness-invariance}
  Well-formedness is preserved by $(\rewrites)$: if $s\rewrites^* s'$ and $s$ is well-formed, then $s'$ is well-formed.\qed
\end{proposition}

Well-formedness prevents the unwanted scenario of having two tokens on
the same wire, and oriented in the same direction
(e.g.~$(e_0\downarrow x)(e_0\downarrow y)$). As shown in the
Proposition~\ref{prop:well-formed-terms-charcterisation}, this
property is in fact stronger.

\begin{proposition}[Full Characterisation of Well-Formed Terms]
  \label{prop:well-formed-terms-charcterisation}
  Let $D$ be a ZX-diagram, and $s\in \cat{tkS}(D)$ be \emph{not}
  well-formed, i.e.~there exists a term $t$ in $s$, and
  $p\in\operatorname{Paths}(D)$ such that $|P(p,t)|\geq2$. Then we can rewrite
  $s\rewrites s'$ such that a term in $s'$ has a product of at least
  two tokens of the form $(e_0,d,\_)$.\qed
\end{proposition}

Although well-formedness prevents products of tokens on the same wire,
it does not guarantee termination: for this we need to
consider polarities along cycles.

\begin{proposition}[Invariant on Cycles]
\label{prop:invariant-cycles}
Let $D$ be a ZX-diagram, and $c\in\operatorname{Cycles}(D)$ a cycle. Let
$t_1, \dots, t_n$ be tokens, and $s$ be a token state such that $t_1...t_n \rewrites^* s$.
Then for every non-null term $t$ in $s$ we have $P(c,t_1...t_n)=P(c,t)$.\qed
\end{proposition}

This proposition tells us that the polarity is preserved inside a
cycle. By requiring the polarity to be $0$, we can show that the token
machine terminates. This property is defined formally in the
following.

\begin{definition}[Cycle-Balanced Token State]
  Let $D$ be a ZX-diagram, and $t$ a term in a token state on $D$.  We
  say that $t$ is \emph{cycle-balanced} if for all cycles
  $c\in\operatorname{Cycles}(D)$ we have $P(c,t)=0$.  We say that a token state
  is \emph{cycle-balanced} if all its terms are cycle-balanced.
\end{definition}

To show that being cycle-balanced implies termination, we need the
following intermediate lemma. This essentially captures the fact that
a token in the diagram comes from some other token that ``traveled''
in the diagram earlier on.

\begin{lemma}[Rewinding]
  \label{lemma:rewinding}
  Let $D$ be a ZX-diagram, and $t$ be a term in a well-formed token
  state on $D$, and such that $t\rewrites^*\sum_i\lambda_it_i$, with
  $(e_n,d,x)\in t_1$. If $t$ is cycle-balanced, then there exists a
  path $p=(e_0,...,e_n)\in\operatorname{Paths}(D)$ such that $e_n$ is
  $d$-oriented in $p$, and $P(p,t)=1$.\qed
\end{lemma}

We can now prove strong-normalization.

\begin{theorem}[Termination of well-formed, cycle-balanced token state]
  \label{thm:termination}
  Let $D$ be a ZX-diagram, and $s\in\cat{tkS}(D)$ be well-formed. The
  token state
  $s$ is strongly normalizing if and only if it is cycle-balanced.\qed
\end{theorem}

Intuitively, this means that tokens inside a cycle will cancel
themselves out if the token state is cycle-balanced. Since cycles are
the only way to have a non-terminating token machine, we are sure that
our machine will always terminate.

\begin{proposition}[Local Confluence]
  \label{prop:local_confluence}
  Let $D$ be a ZX-diagram, and $s\in\cat{tkS}(D)$ be well-formed and
  collision-free. Then, for all $s_1,s_2\in\cat{tkS}(D)$ such that
  $s_1~\setirwer~ s\rewrites s_2$, there exists $s'\in\cat{tkS}(D)$ such
  that $s_1\rewrites^* s'~{}^*\setirwer~ s_2$.\qed
\end{proposition}

\begin{corollary}[Confluence]
  \label{cor:confluence}
  Let $D$ be a ZX-diagram. The rewrite system $\rewrites$ is confluent
  for well-formed, collision-free and cycle-balanced token states.\qed
\end{corollary}

\begin{corollary}[Uniqueness of Normal Forms]
  \label{cor:uniquenessNF}
  Let $D$ be a ZX-diagram. A well-formed and cycle-balanced token
  state admits a unique normal form under the rewrite system
  $\rewrites$.\qed
\end{corollary}

\subsection{Semantics and Structure of Normal Forms}

In this section, we discuss the structure of normal forms, and relate
the system to the standard interpretation presented in
Section~\ref{sec:ZX}.

\begin{proposition}[Single-Token Input]
\label{prop:single-token-input}
Let $D:n\to m$ be a connected $\cat{ZX}$-diagram with $\mathcal I(D)=[a_i]_{0<i\leq n}$ and $\mathcal O(D)=[b_i]_{0<i\leq m}$, $0< k\leq n$ and $x\in\{0,1\}$, such that:\\
\centerline{$\displaystyle\interp{D}\circ(id_{k-1}\otimes \ket x\otimes id_{n-k})=\sum_{q=1}^{2^{m+n-1}}\lambda_q \ketbra{y_{1,q},...,y_{m,q}}{x_{1,q},...,x_{k-1,q},x_{k+1,q},...,x_{n,q}}$}
Then:\hfill$\displaystyle(a_k\downarrow x)\rewrites^*\sum_{q=1}^{2^{m+n-1}}\lambda_q \prod_i(b_i\downarrow y_{i,q})\prod_{i\neq k}(a_i\uparrow x_{i,q})$\hfill\qed
\end{proposition}
This proposition conveys the fact that dropping a single token in state $x$ on wire $a_k$ gives the same semantics as the one obtained from the standard interpretation on the ZX-diagram, with wire $a_k$ connected to the state $\ket x$.

Proposition \ref{prop:single-token-input} can be made more general. However, we first need the following result on ZX-diagrams:

\begin{lemma}[Universality of Connected ZX-Diagrams]
\label{lem:connected-universality}
Let $f:\mathbb C^{2^n}\to\mathbb C^{2^m}$. There exists a \emph{connected} ZX-diagram $D_f:n\to m$ such that $\interp{D_f}=f$.\qed
\end{lemma}

\begin{proposition}[Multi-Token Input]
\label{prop:n-tokens}
Let $D$ be a \emph{connected} ZX-diagram with $\mathcal I(D)=[a_i]_{1\leq i\leq n}$ and $\mathcal O(D)=[b_i]_{1\leq i\leq m}$; with $n\geq 1$.\\
If: $\qquad\displaystyle\interp{D}\circ\left(\sum_{q=1}^{2^n}\lambda_q\ket{x_{1,q},...,x_{n,q}}\right) = \sum_{q=1}^{2^m}\lambda_q'\ket{y_{1,q},...,y_{m,q}}$\hfill~\\
then:$\qquad\displaystyle\sum_{q=1}^{2^n}\lambda_q\prod_{i=1}^n(a_i\downarrow x_{i,q}) \rewrites^*\sum_{q=1}^{2^m}\lambda_q'\prod_{i=1}^m(b_i\downarrow y_{i,q})$\hfill\qed
\end{proposition}

\begin{example}[CNOT]
\label{ex:CNOT}
In the ZX-Calculus, the CNOT-gate (up to some scalar) can be
constructed as follows:
$\interp{
\InputIfFileExists{./figures/cnot-ex-app.tikz}{}{{\color{red}\colorbox{pink}{missing file : cnot-ex-app}}}
}=\frac{1}{\sqrt{2}}\begin{pmatrix}1&0&0&0\\0&1&0&0\\0&0&0&1\\0&0&1&0\end{pmatrix}$

On classical inputs, this gate applies the NOT-gate on the second bit if and only if the first bit is at $1$.
Therefore if we apply the state $\ket{10}$ to it we get $\frac{1}{\sqrt{2}}\ket{11}$.

We demonstrate how the token machine can be used to get this result. Following Proposition~\ref{prop:n-tokens}, we start by initialising the Token Machine in the token state $(a_1\downarrow 1)(a_2\downarrow 0)$, matching the input state $\ket{10}$.

We underline each step that is being rewritten, and take the liberty to sometimes do several rewrites in parallel at the same time.
{\setlength{\mathindent}{0pt}
\begin{align*}
    \textstyle
    \Cline{(a_1\downarrow 1)}(a_2\downarrow 0)
    &\textstyle\rewrites_d  {\color{red}(b_1\downarrow 1)(e_1\downarrow 1)}\Cline{(a_2\downarrow 0)}
    \rewrites_d (b_1\downarrow 1)(e_1\downarrow 1){\color{red}\frac{1}{\sqrt{2}}\Big(\Cline[blue]{(e_3\downarrow 0)} + \Cline[green!60!black]{(e_3\downarrow 1)}\Big)}\\
    &\textstyle\rewrites_d \frac{1}{\sqrt{2}}(b_1\downarrow 1)\Cline{(e_1\downarrow 1)}\Big({\color{blue}(e_2\uparrow 0)(e_4\downarrow 0)} + {\color{green!60!black}(e_2\uparrow 1)(e_4\downarrow 1)}\Big)\\
    &\textstyle\rewrites_d {\color{red}\frac{1}{2}}(b_1\downarrow 1){\color{red}\Big(\Cline[blue]{(e_2\downarrow 0)}-\Cline[green!60!black]{(e_2\downarrow 1)}\Big)}\Big(\Cline[blue]{\Cline[green!60!black]{(e_2\uparrow 0)}}(e_4\downarrow 0) + (e_2\uparrow 1)(e_4\downarrow 1)\Big)\\
    &\textstyle\rewrites_c^2 \frac{1}{2}(b_1\downarrow 1)\Big({\color{blue}(e_4\downarrow 0)} + \big(\Cline[red]{(e_2\downarrow 0)}-\Cline[green!60!black]{(e_2\downarrow 1)}\big)\Cline[red]{\Cline[green!60!black]{(e_2\uparrow 1)}}(e_4\downarrow 1)\Big)\\
    &\textstyle\rewrites_c^2 \frac{1}{2}(b_1\downarrow 1)\Big(\Cline[red]{(e_4\downarrow 0)}-{\color{green!60!black}\Cline[blue]{(e_4\downarrow 1)}}\Big)\\
    &\textstyle\rewrites_d \frac{1}{2\sqrt2}(b_1\downarrow 1)\Big({\color{red}(b_2\downarrow 0)+(b_2\downarrow 1)}-{\color{blue}(b_2\downarrow 0)+(b_2\downarrow 1)}\Big)\\
    &\textstyle\quad= \frac{1}{\sqrt2}(b_1\downarrow 1)(b_2\downarrow 1)
\end{align*}}

The final token state corresponds to $\frac{1}{\sqrt{2}}\ket{11}$, as described by Proposition~\ref{prop:n-tokens}.
Notice that during the run, each invariants presented before holds and that due to confluence we could have rewritten the tokens in any order and still obtain the same result.
\end{example}

This proposition is a direct generalization of the
proposition~\ref{prop:single-token-input}. It shows we can compute the output of a diagram provided a particular input state. We can also recover the semantics of the whole operator by initialising the starting token state in a particular configuration.

\begin{theorem}[Arbitrary Wire Initialisation]
\label{thm:arbitrary-wire-init}
Let $D$ be a connected ZX-diagram, with $\mathcal I(D)=[a_i]_{1\leq i\leq n}$, $\mathcal O(D)=[b_i]_{1\leq i\leq m}$, and $e\in \mathcal E(D)\neq\emptyset$ such that $(e\downarrow x)(e\uparrow x) \rewrites^* t_x$ for $x\in\{0,1\}$ with $t_x$ terminal (the rewriting terminates by Corollary \ref{cor:uniquenessNF}). Then:\\
$\displaystyle\interp{D}=\!\sum_{q=1}^{2^{m+n}}\lambda_q\ketbra{y_{1\!,q}\ldots{}y_{m,q}}{x_{1\!,q}\ldots{}x_{n,q}}\implies
t_0+t_1 = \!\sum_{q=1}^{2^{m+n}}\lambda_q \prod_i(b_i\!\downarrow\!y_{i,q})\prod_{i}(a_i\!\uparrow\!x_{i,q})$\hfill\qed
\end{theorem}

\begin{example}
If we take back the diagram from Example~\ref{ex:CNOT} and decide to initialize
any wire $e$ of the diagram in the state $(e\downarrow 0)(e\uparrow0) + (e\downarrow 1)(e\uparrow 1)$ and apply the rewriting as in Theorem~\ref{thm:arbitrary-wire-init} we in fact end up with the token state $\frac{1}{\sqrt{2}}\bigg((a_1\uparrow 0)(a_2\uparrow 0)(b_1\downarrow 0)(b_2\downarrow 0) + (a_1\uparrow 0)(a_2\uparrow 1)(b_1\downarrow 0)(b_2\downarrow 1) + (a_1\uparrow 1)(a_2\uparrow 0)(b_1\downarrow 1)(b_2\downarrow 1)  + (a_1\uparrow 1)(a_2\uparrow 1)(b_1\downarrow 1)(b_2\downarrow 0) \bigg)$ which matches the actual matrix of the standard interpretation.
\end{example}

\begin{remark}
At this point, it is legitimate to wonder about the benefits of the token machine over the standard interpretation for computing the semantics of a diagram. Let us first notice that when computing the semantics of a diagram à la Theorem \ref{thm:arbitrary-wire-init}, we get in the token state one term per non-null entry in the associated matrix (the one obtained by the standard interpretation).

We can already see that the token-based interpretation can be interesting if the matrix is sparse, the textbook case being $Z^n_n$ whose standard interpretation requires a $2^n\times2^n$ matrix, while the token-based interpretation only requires two terms (each with $2n$ tokens).

Secondly, we can notice that we can "mimic" the standard interpretation with the token machine. Consider a diagram decomposed as a product of slices (tensor product of generators) for the standard interpretation. Then, for the token machine, without going into technical details, we can follow the strategy that consists in moving token through the diagram one slice at a time. This essentially computes the matrix associated with each slice and its composition.

The point of the token machine however, is that it is versatile enough to allow for more original strategies, some of which may have a worst complexity, but also some of which may have a better one.
\end{remark}

\section{Extension to Mixed Processes}
\label{sec:mixed-processes}
The token machine presented so far worked for so-called \emph{pure}
quantum processes i.e.~with no interaction with the environment. To
demonstrate how generic our approach is, we show how to adapt it to
the natural extension of \emph{mixed} processes, represented with
completely positive maps (CPM). This in particular allows us to
represent quantum measurements.

\subsection{ZX-diagrams for Mixed Processes}

The interaction with the environment can be modeled in the ZX-Calculus by adding a
unary generator $\ground$ to the language~\cite{coecke2012environment,carette2019completeness}, that intuitively
enforces the state of the wire to be classical.
We denote with $\cat{ZX}^{\sground}$ the set of diagrams obtained by
adding $\ground$ this generator.

Similar to what is done in quantum computation, the standard
interpretation $\interp{.}^{\sground}$ for $\cat{ZX}^{\sground}$ maps
diagrams to CPMs. If $D\in\cat{ZX}$ we define $\interp{D}^{\sground}$
as $\rho\mapsto \interp{D}^\dagger\circ\rho\circ\interp{D}$, and we
set $\interp{\ground}^{\sground}$ as $\rho\mapsto\tr(\rho)$, where
$\tr(\rho)$ is the trace of $\rho$.

There is a canonical way to map a $\cat{ZX}^{\sground}$-diagram to a
$\cat{ZX}$-diagram in a way that preserves the semantics: the
so-called CPM-construction~\cite{cpmsellinger}. We define the map
(conveniently named) $\operatorname{CPM}$ as the map that preserves
compositions $(\_\circ\_)$ and $(\_\otimes\_)$ and such that:

\medskip

$\forall D\in\cat{ZX},~\ccpm{
\InputIfFileExists{./figures/D-no-name.tikz}{}{{\color{red}\colorbox{pink}{missing file : D-no-name}}}
}=
\InputIfFileExists{./figures/CPM-pure.tikz}{}{{\color{red}\colorbox{pink}{missing file : CPM-pure}}}
$\hfill$\ccpm{\ground}=~
\begin{tikzpicture}
	\begin{pgfonlayer}{nodelayer}
		\node [style=none] (0) at (-0.25, 0.125) {};
		\node [style=none] (1) at (0.25, 0.125) {};
		\node [style=none] (2) at (0, -0.125) {};
	\end{pgfonlayer}
	\begin{pgfonlayer}{edgelayer}
		\draw [bend right=90, looseness=1.75] (0.center) to (1.center);
	\end{pgfonlayer}
\end{tikzpicture}
\colorbox{pink}{missing file : cup}}}
\hfill~$

\medskip
\noindent
Where $[D]^{\operatorname{cj}}$ is $D$ where every angle $\alpha$ has been changed to $-\alpha$.

With respect to what happens to edge labels, notice that every edge in
$D$ can be mapped to 2 edges in $\operatorname{CPM}(D)$. We propose
that label $e$ induces label $e$ in the first copy, and $\overline e$
in the second, e.g, for the identity diagram:
$~~
\begin{tikzpicture}
	\begin{pgfonlayer}{nodelayer}
		\node [style=none] (0) at (0, 0.25) {};
		\node [style=none] (1) at (0, -0.25) {};
	\end{pgfonlayer}
	\begin{pgfonlayer}{edgelayer}
		\draw (0.center) to (1.center);
	\end{pgfonlayer}
\end{tikzpicture}
\colorbox{pink}{missing file : id}}}
\phantom{.\!}_{\color{gray}e_0}~~\longmapsto~~

\begin{tikzpicture}
	\begin{pgfonlayer}{nodelayer}
		\node [style=none] (0) at (0, 0.25) {};
		\node [style=none] (1) at (0, -0.25) {};
	\end{pgfonlayer}
	\begin{pgfonlayer}{edgelayer}
		\draw (0.center) to (1.center);
	\end{pgfonlayer}
\end{tikzpicture}
\colorbox{pink}{missing file : id}}}
\phantom{.\!}_{\color{gray}{e_0}}~
\begin{tikzpicture}
	\begin{pgfonlayer}{nodelayer}
		\node [style=none] (0) at (0, 0.25) {};
		\node [style=none] (1) at (0, -0.25) {};
	\end{pgfonlayer}
	\begin{pgfonlayer}{edgelayer}
		\draw (0.center) to (1.center);
	\end{pgfonlayer}
\end{tikzpicture}
\colorbox{pink}{missing file : id}}}
\phantom{.\!}_{\color{gray}{\overline{e_0}}}$

In the general ZX-Calculus, it has been shown that the axiomatization itself could be extended to a complete one by adding only 4 axioms \cite{carette2019completeness}.

\vspace*{-1em}
\noindent\begin{minipage}{0.65\columnwidth}
\begin{example}
A $\cat{ZX}^{\sground}$-diagram and its associated CPM construction is
shown on the right (without names on the wires for simplicity).
\end{example}
\end{minipage}
\hfill\scalebox{0.8}{
\InputIfFileExists{./figures/example-CPM.tikz}{}{{\color{red}\colorbox{pink}{missing file : example-CPM}}}
}
\vspace*{-1em}

\subsection{Token Machine for Mixed Processes}

We now aim to adapt the token machine to $\cat{ZX}^{\sground}$, the
formalism for completely positive maps. To do so we give an additional state to
each token to mimic the evolution of two token on $\operatorname{CPM}(D)$.

\begin{definition}
Let $D$ be a ZX-diagram. A $\sground$-token is a quadruplet $(p,d,x,y)\in \mathcal E(D)\times\{\downarrow,\uparrow\}\times \{0,1\}\times \{0,1\}$.
We denote the set of $\sground$-tokens on $D$ by $\cat{tk^{\sground}}(D)$.
A $\sground$-token-state is then a multivariate polynomial over $\mathbb C$, evaluated in $\cat{tk^{\sground}}(D)$.
We denote the set of $\sground$-token-states on $D$ by $\cat{tkS^{\sground}}(D)$
\end{definition}

In other words, the difference with the previous machine is that tokens here have an additional state (e.g. $y$ in $(a\downarrow x,y)$). The rewrite rules are given in appendix in Table \ref{tab:cpm-rules}.
\begin{table}[!htb]
\begin{align*}
    
\InputIfFileExists{./figures/id-nw.tikz}{}{{\color{red}\colorbox{pink}{missing file : id-nw}}}
&\quad&(e_0\downarrow x, y)&{}(e_0\uparrow x', y')\rewrites_c \delta_{x,x'}\delta_{y,y'}\tag{Collision}\\[0.5em]

\InputIfFileExists{./figures/cup-nw.tikz}{}{{\color{red}\colorbox{pink}{missing file : cup-nw}}}
&&(e_b\downarrow x, y) &{}\rewrites_d (e_{\neg b}\uparrow x, y)\tag{$
\InputIfFileExists{./figures/cup.tikz}{}{{\color{red}\colorbox{pink}{missing file : cup}}}
$-diffusion}\\

\InputIfFileExists{./figures/cap-nw.tikz}{}{{\color{red}\colorbox{pink}{missing file : cap-nw}}}
&&(e_b\uparrow x,y) &{}\rewrites_d (e_{\neg b}\downarrow x,y)\tag{
\InputIfFileExists{./figures/cap.tikz}{}{{\color{red}\colorbox{pink}{missing file : cap}}}
-diffusion}\\
&&(e_k\downarrow x,y)&{}\rewrites_d e^{i\alpha(x-y)}
        \prod_{j\not=k} (e_j\uparrow x, y)\prod_j (e_j'\downarrow x, y)
\\
&&(e_k'\uparrow x, y)&{}\rewrites_d e^{i\alpha(x-y)}
        \prod_{j} (e_j\uparrow x, y)\prod_{j\neq k} (e_j'\downarrow x, y)
\hspace{10ex}~
\\
&&(e_0\downarrow x,y)&{}\rewrites_d \frac 1{2}\sum_{z,z'\in\{0,1\}}(-1)^{xz+yz'}(e_1\downarrow z,z')
\\
&&(e_1\uparrow x,y)&{}\rewrites_d \frac 1{2}\sum_{z,z'\in\{0,1\}}(-1)^{xz+yz'}(e_0\uparrow z,z')
\\
    \overset{\text{\scriptsize\color{gray}$e_0$}}{\ground}&&(e_0\downarrow x,y)&{}\rewrites_d \delta_{x,y}\tag{Trace-Out}\\
\end{align*}%
\begin{textblock}{1}(0.7,-3)
\InputIfFileExists{./figures/gn-alpha-nw.tikz}{}{{\color{red}\colorbox{pink}{missing file : gn-alpha-nw}}}
\end{textblock}%
\begin{textblock}{1}(0.8,-2)
\InputIfFileExists{./figures/H-nw.tikz}{}{{\color{red}\colorbox{pink}{missing file : H-nw}}}
\end{textblock}%
\begin{textblock}{2}(8.7,-1.8)(
\InputIfFileExists{./figures/H.tikz}{}{{\color{red}\colorbox{pink}{missing file : H}}}
-Diffusion)\end{textblock}%
\begin{textblock}{2}(8.7,-2.8)(
\InputIfFileExists{./figures/gn.tikz}{}{{\color{red}\colorbox{pink}{missing file : gn}}}
-Diffusion)\end{textblock}
    \caption{The rewrite rules for $\cpm$, where $\delta$ is the Kronecker delta.}
    \label{tab:cpm-rules}
\end{table}

It is possible to link this formalism back to the pure
token-states, using the existing CPM construction
for ZX-diagrams. We extend this map by
$\operatorname{CPM}:\cat{tkS^{\sground}}(D)\to\cat{tkS}(\operatorname{CPM}(D))$,
defined as: $\displaystyle \sum_{q=1}^{2^{m+n}} \lambda_q\prod_j(p_j,d_j,x_{j,q},y_{j,q})
\mapsto \sum_{q=1} \lambda_q\prod_j({p_j},d_j,x_{j,q})(\overline{p_j},d_j,y_{j,q})$

\smallskip
\noindent
Since $\operatorname{CPM}(D)$ can be seen as two copies of $D$ where $\ground$ is replaced by 
\begin{tikzpicture}
	\begin{pgfonlayer}{nodelayer}
		\node [style=none] (0) at (-0.25, 0.125) {};
		\node [style=none] (1) at (0.25, 0.125) {};
		\node [style=none] (2) at (0, -0.125) {};
	\end{pgfonlayer}
	\begin{pgfonlayer}{edgelayer}
		\draw [bend right=90, looseness=1.75] (0.center) to (1.center);
	\end{pgfonlayer}
\end{tikzpicture}
\colorbox{pink}{missing file : cup}}}
, each token in $D$ corresponds to two tokens in $\operatorname{CPM}(D)$, at the same spot but in the two copies of $D$. The two states $x$ and $y$ represent the states of the two corresponding tokens.

We can then show that this rewriting system is consistent:
\begin{theorem}
\label{thm:simulation-cpm}
Let $D$ be a $\cat{ZX}^{\sground}$-diagram, and $t_1,t_2\in
\cat{tkS^{\sground}}(D)$. Then whenever
$t_1\cpm t_2$ we have $\operatorname{CPM}(t_1)\rewrites^{\{1,2\}}\operatorname{CPM}(t_2)$.\qed
\end{theorem}

The notions of polarity, well-formedness and cycle-balancedness can be adapted, and  we get strong normalization (Theorem~\ref{thm:termination}), confluence (Corollary~\ref{cor:confluence}), and uniqueness of normal forms (Corollary~\ref{cor:uniquenessNF}) for well-formed and cycle-balanced token states.

\section{Conclusion and Future Work}
\label{sec:future-work}

In this paper, we presented a novel \emph{particle-style} semantics
for ZX-Calculus. Based on a token-machine automaton, it emphasizes the
\emph{asynchronicity} and \emph{non-orientation} of the computational
content of a ZX-diagram. Compared to existing token-based semantics of
quantum computation such as~\cite{goisync}, our proposal furthermore
support decentralized tokens where the position of a token can be in
superposition.

As quantum circuits can be mapped to ZX-diagrams, our token machines
induce a notion of asynchronicity for quantum circuits. This contrasts
with the notion of token machine defined in \cite{goisync} where some
form of synchronicity is enforced.

Our token machines give us a new way to look at how a ZX-diagram
computes with a more local, operational approach. This could lead to
extensions of the ZX-Calculus with more expressive logical and
computational constructs, such as
recursion.

As a final remark, we notice that this formalism naturally extends to other graphical languages for qubit quantum computation, and even for tensor networks. It suffices to adapt the diffusion rewriting steps to the generators at hand, which is always possible in the setting of finite dimensional Hilbert spaces, and if needs be to adapt the states in tokens to the dimension of the wire they go through (e.g. if a wire in a tensor network is of dimension 4, the state spans $\{0,1,2,3\}$).

\bibliography{bibli.bib}
\section{Rewrite Rules for Mixed Processes}
~
\begin{table}[!htb]
\begin{align*}
    
\InputIfFileExists{./figures/id-nw.tikz}{}{{\color{red}\colorbox{pink}{missing file : id-nw}}}
&\quad&(e_0\downarrow x, y)&{}(e_0\uparrow x', y')\rewrites_c \delta_{x,x'}\delta_{y,y'}\tag{Collision}\\[0.5em]

\InputIfFileExists{./figures/cup-nw.tikz}{}{{\color{red}\colorbox{pink}{missing file : cup-nw}}}
&&(e_b\downarrow x, y) &{}\rewrites_d (e_{\neg b}\uparrow x, y)\tag{$
\InputIfFileExists{./figures/cup.tikz}{}{{\color{red}\colorbox{pink}{missing file : cup}}}
$-diffusion}\\

\InputIfFileExists{./figures/cap-nw.tikz}{}{{\color{red}\colorbox{pink}{missing file : cap-nw}}}
&&(e_b\uparrow x,y) &{}\rewrites_d (e_{\neg b}\downarrow x,y)\tag{
\InputIfFileExists{./figures/cap.tikz}{}{{\color{red}\colorbox{pink}{missing file : cap}}}
-diffusion}\\
&&(e_k\downarrow x,y)&{}\rewrites_d e^{i\alpha(x-y)}
        \prod_{j\not=k} (e_j\uparrow x, y)\prod_j (e_j'\downarrow x, y)
\\
&&(e_k'\uparrow x, y)&{}\rewrites_d e^{i\alpha(x-y)}
        \prod_{j} (e_j\uparrow x, y)\prod_{j\neq k} (e_j'\downarrow x, y)
\hspace{10ex}~
\\
&&(e_0\downarrow x,y)&{}\rewrites_d \frac 1{2}\sum_{z,z'\in\{0,1\}}(-1)^{xz+yz'}(e_1\downarrow z,z')
\\
&&(e_1\uparrow x,y)&{}\rewrites_d \frac 1{2}\sum_{z,z'\in\{0,1\}}(-1)^{xz+yz'}(e_0\uparrow z,z')
\\
    \overset{\text{\scriptsize\color{gray}$e_0$}}{\ground}&&(e_0\downarrow x,y)&{}\rewrites_d \delta_{x,y}\tag{Trace-Out}\\
\end{align*}%
\begin{textblock}{1}(0.7,-3)
\InputIfFileExists{./figures/gn-alpha-nw.tikz}{}{{\color{red}\colorbox{pink}{missing file : gn-alpha-nw}}}
\end{textblock}%
\begin{textblock}{1}(0.8,-2)
\InputIfFileExists{./figures/H-nw.tikz}{}{{\color{red}\colorbox{pink}{missing file : H-nw}}}
\end{textblock}%
\begin{textblock}{2}(8.7,-1.8)(
\InputIfFileExists{./figures/H.tikz}{}{{\color{red}\colorbox{pink}{missing file : H}}}
-Diffusion)\end{textblock}%
\begin{textblock}{2}(8.7,-2.8)(
\InputIfFileExists{./figures/gn.tikz}{}{{\color{red}\colorbox{pink}{missing file : gn}}}
-Diffusion)\end{textblock}
    \caption{The rewrite rules for $\cpm$, where $\delta$ is the Kronecker delta.}
    \label{tab:cpm-rules}
\end{table}

\section{Proofs of Section~\ref{sec:token-machine}}

\begin{proof}[Proof of Proposition \ref{prop:well-formedness-invariance}]

Let $D$ be a ZX-diagram, and $s$ be a well-formed token state on $D$. Consider a rewrite $s\rewrites s'$. We want to show that for all paths $p$ in $D$, if $P(p,t)\in\{-1,0,1\}$ for all terms $t$ of $s$, then $P(p,t')\in\{-1,0,1\}$ for all terms $t'$ in $s'$.

Let $t$ be a term of $s$, and $e_0$ be the edge where a rewriting occurs. If the rewriting does not affect $t$, then the well-formedness of $t$ obviously holds. If it does, and $t\rewrites_{c,d} \sum_qt_q$, we have to check two cases:
\begin{itemize}
    \item Collision: let $p\in\operatorname{Paths}(D)$. If no tokens remain in the term $t_q$, then $P(p,t_q)=0$. Otherwise:
    \begin{itemize}
        \item if $e_0\notin p$, then $P(p,t_q)=P(p,t)$
        \item if $e_0\in p$, then $P(p,t_q)=P(p,t)+1-1$ because the two tokens have alternating polarity
    \end{itemize}
    \item Diffusion: let $p\in\operatorname{Paths}(D)$, and $(e_0,d, x)\rewrites_d \sum_q\lambda_q\prod_{i\in S}(e_i,d _i,x_{i,q})$ (this captures all possible diffusion rules).
    \begin{itemize}
        \item if $e_0\notin p$ and $\forall i, e_i\notin p$, then $P(p,t_q)=P(p,t)$
        \item if $e_0\in p$ and $\exists k\in S, e_k\in p$, then $\forall i\neq k, e_i\notin p$, because the generator can only be passed through once by the path $p$. We have $P(p,(e_0,d, x)) = P(p,(e_k,d_k,x_{k,q}))$ by the definition of orientation in a path, which means that $\forall q, P(p,t_q)=P(p,t)$
        \item if $e_0\in p$ and $\forall i, e_i\notin p$, then, either $i)$ $p$ ends with $e_0$ and $e_0$ is $d$-oriented in $p$, or $ii)$ $p$ starts with $e_0$ and $e_0$ is $\neg d$-oriented in $p$. In both cases, since that $p\setminus\{e_0\}$ is still a path, we have $P(p\setminus\{e_0\},t)\in\{-1,0,1\}$ and since $P(p,t_q)=P(p\setminus\{e_0\},t)$, we deduce that $t_q$ is still well-formed
        \item if $e_0\notin p$ but $\exists k\in S, e_k\in p$, either $e_k$ is an extremity of $p$, or $\exists k',e_{k'}\in p$. In the latter case, the tokens in $e_k$ and $e_{k'}$ will have alternating polarity in $p$, so $\forall q,P(p,t_q)=P(p,t)+1-1$. In the first case, we can show in a way similar to the previous point, that $P(p,t_q)=P(p\setminus\{e_k\},t)\in\{-1,0,1\}$
    \end{itemize}
\end{itemize}
\end{proof}

\begin{proof}[Proof of Proposition \ref{prop:well-formed-terms-charcterisation}]
Let $t$ be a term in $s$, and $p=(e_0,...,e_n)$ such that $P(p,t)\geq2$ (if $P(p,t)\leq -2$, can can simply take the reversed path). We can show that we can rewrite $t$ into a token state with term $t' = (e_i,d,\_)(e_i,d,\_)t''$. We do so by induction on $n=|p|-1$.

If $n=0$, we have a path constituted of one edge, such that $|P(p,t)|\geq2$. Even after doing all possible collisions, we are left with $|P(p,t)|$ tokens on $e_0$, and oriented accordingly.

For $n+1$, we look at $e_0$, build $p':=(e_1,...,e_n)$, and distinguish four cases. If there is no token on $e_0$, we have $P(p',t)=P(p,t)$, so the result is true by induction hypothesis on $p'$. If we have a product of at least two tokens going in the same direction, the result is directly true. If we have exactly one token going in each direction, we apply the collision rules, and still have $P(p',t)=P(p,t)$, so the result is true by induction hypothesis on $p'$. Finally, if we have exactly one token $(e_0,d,\_)$ on $e_0$, either $e_0$ is not $d$-oriented, in which case $P(p',t)=P(p,t)+1$, or $e_0$ is $d$-oriented, in which case the adequate diffusion rule on $(e_0,d,\_)$ will rewrite $t\rewrites \sum_qt_q$ with $P(p',t_q)=P(p,t)$.
\end{proof}

\begin{proof}[Proof of Proposition~\ref{prop:invariant-cycles}]
The proof can be adapted from the proof of Proposition~\ref{prop:well-formedness-invariance}, by forgetting the cases related to the extremity of the paths, as well as the null terms (which can arise from collisions). It can then be observed that the quantity $P$ in this simplified setting is more than bounded to $\{-1,0,1\}$, but preserved.
\end{proof}

\begin{proof}[Proof of Lemma~\ref{lemma:rewinding}]
We reason by induction on the length $k$ of the rewrite that leads from $t$ to $\sum_i\lambda_it_i$.

If $k=0$, we have $(e_n,d,x)\in t$, so the path $p:=(e_n)$ is sufficient.

For $k+1$, suppose $t\rewrites \sum_i\lambda_it_i$, and $t_1\rewrites^k \sum_j\lambda'_jt'_j$ (hence $t\rewrites^{k+1}\sum_{i\neq1}\lambda_it_i+\sum_j\lambda'_jt'_j$), with $(e_n,d,x)\in t'_1$. By induction hypothesis, there is $p=(e_0,...,e_n)$ such that $P(p,t_1)=1$. We now need to look at the first rewrite from $t$.
\begin{itemize}
    \item if the rewrite concerns a generator not in $p$, then $P(p,t)=P(p,t_1)=1$
    \item if the rewrite is a collision, then $P(p,t)=P(p,t_1)=1$
    \item if the rewrite is $(e,d_e,x_e)\rewrites \sum_q\lambda_q\prod_i(e'_i,d_i,x_{i,q})$
    \begin{itemize}
        \item if $e\in p$ and $e'_1\in p$, then $P(p,t)=P(p,t_1)=1$
        \item if $e'_1\in p$ and $e'_2\in p$, then $P(p,t)=P(p,t_1)-1+1=1$
        \item the case $e\in p$ and $\forall i,~e'_i\notin p$ is impossible:
        \begin{itemize}
            \item if $e$ is not $d_e$-oriented in $p$, it means $e=e_0$, hence $P((e_1,...,e_n),t)=P(p,t)+1=2$ which is forbidden by well-formedness
            \item if $e$ is $d_e$-oriented in $p$, it means $e=e_n$, which would imply that $P(p,t_1)=0$
        \end{itemize}
        \item if $e\notin p$ and $e'_1\in p$ and $\forall i\neq1, e'_i\notin p$, then $P(e::p,t)=P(p,t_1)=1$, since well-formedness prevents the otherwise possible situation $P(e::p,t)=P(p,t_1)+1=2$. However, $e::p$ may not be a path anymore. If $c=(e,e_0,...,e_\ell)$ forms a cycle, then, since $P(c,t)=0$, we can simply keep the path $p':=(e_{\ell+1},...,e_n)$ with $P(p',t)=1$
    \end{itemize}
    
\end{itemize}
\end{proof}

\begin{proof}[Proof of Theorem~\ref{thm:termination}]
$[\Rightarrow]$: Suppose $\exists c\in\operatorname{Cycles}(D)$ and $t$ a term of $s$ such that $P(c,t)\neq 0$. By well-formedness, $P(c,t)\in\{-1,1\}$. Any terminal term $t'$ has $P(c,t')=0$, so by preservation of the quantity $P(c,\_)$, $t$ (and henceforth $s$) cannot terminate.

$[\Leftarrow]$: We are going to show for the reciprocal that, if $t$ is well-formed, and if the constraint $P(c,t)=0$ is verified for every cycle $c$, then any generator in the diagram can be visited at most once. More precisely, we show that if a generator is visited in a term $t$, then it cannot be visited anymore in all the terms derived from $t$. However, the same generator can be visited once for each superposed term (e.g.~once in $t_1$ and once in $t_2$ for the token state $t_1+t_2$).\\
Consider an edge $e$ with token exiting generator $g$ in the term $t$.
Suppose, by reductio ad absurdum, that a token will visit $g$ again in
$t'$ (obtained from $t$), by edge $e_n$ with orientation $d$. By Lemma
\ref{lemma:rewinding}, there exists a path $p=(e_0,...,e_n)$ such that
$P(p,t)=1$ and $e_n$ is $d$-oriented. Since $e\notin p$ (we would not
have a path then), then $p':=(e_0,...,e_n,e)$ is a path (or possibly a
cycle) such that $P(p',t)=2$. This is forbidden by well-formedness.
Hence, every generator can be visited at most once. As a consequence,
the lexicographic order $(\#g,\#tk)$ (where $\#g$ is the number of
non-visited generators in the diagram, and $\#tk$ the number of tokens
in the diagram) strictly reduces with each rewrite. This finishes the
proof of termination.
\end{proof}

\begin{proof}[Proof of Proposition~\ref{prop:local_confluence}]
We are going to reason on every possible pair of rewrite rules that can be applied from a single token state $s$. Notice first, that if the two rules are applied on two different terms of $s$, such that the rewriting of a term creates a copy of the other, they obviously commute, so $\begin{array}{c@{~}c@{~}c}
    s & \rewrites & s_2 \\
    \downrewrites &  & \downrewrites\\
    s_1 & \rewrites & s'
\end{array}$.\\
In the case where $s=\alpha t+\beta t_1 +s_0$ such that $t_1\rewrites s'$ and $t\rewrites \sum_i\lambda_it_i$, we have:
$$\begin{array}{c@{~}c@{~}c@{~}c@{~}c}
     & \nerewrites & \alpha t+\beta s'+s_0 & \rewrites & \sum_i\alpha\lambda_it_i+\beta s'+s_0\\
    s &  &  &  & \downrewrites \\
     & \serewrites & (\alpha\lambda_1+\beta)t_1+ \sum_{i\neq1}\alpha\lambda_it_i+s_0 & \rewrites & (\alpha\lambda_1+\beta)s'+ \sum_{i\neq1}\alpha\lambda_it_i+s_0
\end{array}$$

Then, we can, in the following, focus on pairs of rules applied on the same term.\\
The term we focus on is obviously collision-free, by hypothesis and by
preservation of collision-freeness by $\rewrites$.

Suppose the two rewrites are applied on tokens at positions $e$ and $e'$. We may reason using the distance between the two edges.
\begin{itemize}
    \item the case $d(e,e')=0$ would imply a collision, which is impossible by collision-freeness
    \item if $d(e,e')\geq3$, the two rules still don't interfere, they commute (up to collisions which do not change the result)
    \item if $d(e,e')=2$, there will be common collisions (i.e.~collisions between tokens created by each of the diffusions), however, the order of application of the rules will not change the bits in the tokens we will apply a collision on, so the result holds
    \item if $d(e,e')=1$, then the two tokens have to point to the
      same generator. If they didn't, $(e,e')$ would form a path such
      that $|P((e,e'),t)|=2$ which is forbidden by well-formedness. We
      can then show the property for all generators:

    Case $
\InputIfFileExists{./figures/cup-nw.tikz}{}{{\color{red}\colorbox{pink}{missing file : cup-nw}}}
$.

    $$\begin{array}{c@{~}c@{~}c}
        (e_0\downarrow x)(e_1\downarrow x') & \rewrites_d & (e_1\uparrow x)(e_1\downarrow x') \\
        \downrewrites_d &  & \downrewrites_c\\
        (e_0\downarrow x)(e_0\uparrow x') & \rewrites_c & \braket{x}{x'}
    \end{array}$$

    Case $
\InputIfFileExists{./figures/cap-nw.tikz}{}{{\color{red}\colorbox{pink}{missing file : cap-nw}}}
$: similar.

    Case $
\InputIfFileExists{./figures/gn-alpha-nw.tikz}{}{{\color{red}\colorbox{pink}{missing file : gn-alpha-nw}}}
$.

    $$\scalebox{.9}{$\begin{array}{c@{~}c@{~}c}
        e^{i\alpha x}\prod_{i\neq1}(e_i\uparrow x)\prod_{i}(e'_i\downarrow x)(e'_1\uparrow x') &\serewrites_c&\\
        \uprewrites_d &  & \braket{x}{x'}e^{i\alpha x}\prod_{i\neq1}(e_i\uparrow x)\prod_{i\neq 1}(e'_i\downarrow x)\\
        (e_1\downarrow x)(e'_1\uparrow x') &  & |\,| \\
        \downrewrites_d &  & \braket{x}{x'}e^{i\alpha x'}\prod_{i\neq1}(e_i\uparrow x')\prod_{i\neq 1}(e'_i\downarrow x')\\
        e^{i\alpha x'}\prod_{i}(e_i\uparrow x')\prod_{i\neq1}(e'_i\downarrow x)(e_1\downarrow x) & \nerewrites_c &
    \end{array}$}$$
    
    Case $
\InputIfFileExists{./figures/H-nw.tikz}{}{{\color{red}\colorbox{pink}{missing file : H-nw}}}
$.

    $$\scalebox{.9}{$\begin{array}{c@{~}c@{~}c}
        \frac1{\sqrt2}\left((-1)^x(e_1\downarrow x)(e_1\uparrow x') + (e_1\downarrow \neg x)(e_1\uparrow x')\right) &\serewrites_c^2&\\
        \uprewrites_d &  & \frac1{\sqrt2}\left((-1)^x\braket{x}{x'} + \braket{\neg x}{x'}\right)\\
        (e_0\downarrow x)(e_1\uparrow x') &  & |\,| \\
        \downrewrites_d &  & \frac1{\sqrt2}\left((-1)^{x'}\braket{x}{x'} + \braket{x}{\neg x'}\right)\\
        \frac1{\sqrt2}\left((-1)^{x'}(e_0\downarrow x)(e_0\uparrow x') + (e_0\downarrow x)(e_0\uparrow \neg x')\right) & \nerewrites_c^2 &
    \end{array}$}$$
\end{itemize}
\end{proof}

\begin{proof}[Proof of Proposition~\ref{prop:single-token-input}]
Let us first notice that, using the map/state duality, we have 
$(a_k\downarrow x)\rewrites^*\sum\limits_{q=1}^{2^{m+n-1}}\lambda_q \prod\limits_i(b_i\downarrow y_{i,q})\prod\limits_{i\neq k}(a_i\uparrow x_{i,q})$ in $D$ iff we have $(a_k\downarrow x)\rewrites^*\sum\limits_{q=1}^{2^{m+n-1}}\lambda_q \prod\limits_i(b_i\downarrow y_{i,q})\prod\limits_{i\neq k}(a_i'\downarrow x_{i,q})$ in $D'$ where $
\InputIfFileExists{./figures/Dp-from-D.tikz}{}{{\color{red}\colorbox{pink}{missing file : Dp-from-D}}}
$. Hence, we can, w.l.o.g.~consider in the following that $n=1$. We also notice that thanks to the confluence of the rewrite system, we can consider diagrams up to "topological deformations", and hence ignore cups and caps.

We then proceed by induction on the number $N$ of ``non-wire generators'' (i.e.~Z-spider, X-spiders and H-gates) of $D$, using the fact that the diagram is connected:

If $N=0$, then $D = 
\begin{tikzpicture}
	\begin{pgfonlayer}{nodelayer}
		\node [style=none] (0) at (0, 0.25) {};
		\node [style=none] (1) at (0, -0.25) {};
	\end{pgfonlayer}
	\begin{pgfonlayer}{edgelayer}
		\draw (0.center) to (1.center);
	\end{pgfonlayer}
\end{tikzpicture}
\colorbox{pink}{missing file : id}}}
$, where the result is obvious.

If $N=1$, then $D\in\left\lbrace,
\begin{tikzpicture}
	\begin{pgfonlayer}{nodelayer}
		\node [style=none] (0) at (0, 0.25) {};
		\node [style=none] (1) at (0, -0.25) {};
		\node [style={{H box}}] (2) at (0, 0) {};
	\end{pgfonlayer}
	\begin{pgfonlayer}{edgelayer}
		\draw (0.center) to (1.center);
	\end{pgfonlayer}
\end{tikzpicture}
\colorbox{pink}{missing file : H}}}
,
\InputIfFileExists{./figures/gn-1-n.tikz}{}{{\color{red}\colorbox{pink}{missing file : gn-1-n}}}
,
\InputIfFileExists{./figures/rn-1-n.tikz}{}{{\color{red}\colorbox{pink}{missing file : rn-1-n}}}
\right\rbrace$. The result in this base case is then a straightforward verification (self-loops in green and red nodes simply give rise to collisions that are handled as expected).

For $N+1$, there exists $D'$ with $N$ non-wire generators and such that $$D\in\left\lbrace
\InputIfFileExists{./figures/IH-H-top.tikz}{}{{\color{red}\colorbox{pink}{missing file : IH-H-top}}}
,
\InputIfFileExists{./figures/IH-gn-top.tikz}{}{{\color{red}\colorbox{pink}{missing file : IH-gn-top}}}
,
\InputIfFileExists{./figures/IH-rn-top.tikz}{}{{\color{red}\colorbox{pink}{missing file : IH-rn-top}}}
\right\rbrace$$
(we should actually take into account the self loops, but they do not change the result).
Let us look at the first two cases, since the last one can be induced by composition.

If $D=
\InputIfFileExists{./figures/IH-H-top-nw.tikz}{}{{\color{red}\colorbox{pink}{missing file : IH-H-top-nw}}}
$, then $D'$ is necessarily connected, by connectivity of $D$. Then:
\begin{align*}
(a\downarrow x)
&\rewrites\frac{(-1)^x}{\sqrt2}(a'\downarrow x)+\frac1{\sqrt2}(a'\downarrow \neg x)\\
&\rewrites^* \frac{(-1)^x}{\sqrt2}\sum_{q=1}^{2^m}\lambda_{q}\prod_{i=1}^m(b_i\downarrow y_{i,q}) + \frac1{\sqrt2}\sum_{q=1}^{2^m}\lambda'_{q}\prod_{i=1}^m(b_i\downarrow y_{i,q})\\
&\quad = \sum_{q=1}^{2^m}\frac{\lambda'_q+(-1)^x\lambda_q}{\sqrt2}\prod_{i=1}^m(b_i\downarrow y_{i,q})
\end{align*}
where by induction hypothesis 
$$\interp{D'}\ket{x}=\sum_{q=1}^{2^{m}}\lambda_{q}\ket{y_{1,q},...,y_{m,q}}$$
and 
$$\interp{D'}\ket{\neg
  x}=\sum_{q=1}^{2^{m}}\lambda'_{q}\ket{y_{1,q},...,y_{m,q}}$$
so:
\begin{align*}
    \interp{D}\ket x &= \interp{D'\circ H}\ket x = \interp{D'}\circ\interp{H}\ket x
    = \interp{D'}\circ\left(\frac{(-1)^x}{\sqrt2}\ket x+\frac1{\sqrt2}\ket{\neg x}\right)\\
    &= \frac{(-1)^x}{\sqrt2}\interp{D'}\ket x+\frac1{\sqrt2}\interp{D'}\ket{\neg x}
    = \sum_{q=1}^{2^{m}} \frac{\lambda'_{q}+(-1)^x\lambda_{q}}{\sqrt2}\ket{y_{1,q},...,y_{m,q}}
\end{align*}
which is the expected result.

Now, if $D=
\InputIfFileExists{./figures/IH-gn-top.tikz}{}{{\color{red}\colorbox{pink}{missing file : IH-gn-top}}}
$, we can decompose $D'$ in its connected components:
$$D=
\InputIfFileExists{./figures/IH-gn-top.tikz}{}{{\color{red}\colorbox{pink}{missing file : IH-gn-top}}}
=
\InputIfFileExists{./figures/IH-gn-decomp-nw.tikz}{}{{\color{red}\colorbox{pink}{missing file : IH-gn-decomp-nw}}}
$$
with $D_i$ connected. Then:

\begin{align*}
    (a\downarrow x)
    &\rewrites e^{i\alpha x}\prod_\ell\prod_i(a_{\ell,i}\downarrow x)
\\
    &\rewrites^* e^{i\alpha x}\prod_\ell \left(\sum_{q=1}^{2^{m_{\ell}+n_\ell-1}}\lambda_{q,\ell}\prod_{i\neq1}(a_{\ell,i}\downarrow x)(a_{\ell,i}\uparrow x_{\ell,i,q})\prod_i(b_{\ell,i}\downarrow y_{\ell,i,q})\right)\\
    &\rewrites^* e^{i\alpha x}\prod_\ell
      \left(\sum_{q=1}^{2^{m_{\ell}+n_\ell-1}}\lambda_{q,\ell}\delta_{x,x_{\ell,i,q}}\prod_i(b_{\ell,i}\downarrow
      y_{\ell,i,q})\right) \\
&= e^{i\alpha x}\prod_\ell \left(\sum_{q=1}^{2^{m_{\ell}}}\lambda'_{q,\ell}\prod_i(b_{\ell,i}\downarrow y_{\ell,i,q})\right)\\
    &= e^{i\alpha x}\sum_{q_1=1}^{2^{m_1}}...\sum_{q_k=1}^{2^{m_k}}\lambda'_{q_1,1}...\lambda'_{q_k,k}\prod_i(b_{1,i}\downarrow y_{1,i,q_1})...\prod_i(b_{k,i}\downarrow y_{k,i,q_k})
  \\   
  &= \sum_{q=1}^{2^{m}}\lambda'_q\prod_i(b_i\downarrow y_{i,q})
\end{align*}
where the first is the diffusion through a Z-spider, and the second set of rewrites is the induction hypothesis applied to each connected component.

\begin{align*}
    \interp{D}\ket x =&
    \interp{(D_1\otimes...\otimes D_k)\circ Z_k^1(\alpha)}\ket x
    = (\interp{D_1}\otimes...\otimes \interp{D_k})\circ \interp{Z_k^1(\alpha)}\ket x\\
    &= e^{i\alpha x}(\interp{D_1}\otimes...\otimes \interp{D_k})\circ \ket{x,...,x}
    = e^{i\alpha x}\interp{D_1}\ket{x,...,x}\otimes...\otimes \interp{D_k}\ket{x,...,x}\\
    &=e^{i\alpha x}\left(\sum_{q_1}^{2^{m_1+n_1-1}}\lambda_{q_1,1}\ket{y_{1,1,q_1},...,y_{1,m_1,q_1}}\braket{x_{1,2,q_1},...,x_{1,n_1,q_1}}{x,...,x}\right)\otimes\\
    \tag*{$\displaystyle...\otimes \left(\sum_{q_k}^{2^{m_k+n_k-1}}\lambda_{q_k,k}\ket{y_{k,1,q_1},...,y_{k,m_1,q_k}}\braket{x_{k,2,q_k},...,x_{k,n_1,q_k}}{x,...,x}\right)$}\\
    &=e^{i\alpha x}\left(\sum_{q_1}^{2^{m_1+n_1-1}}\lambda_{q_1,1}\prod_i\delta_{x,x_{1,i,q_1}}\ket{y_{1,1,q_1},...,y_{1,m_1,q_1}}\right)\otimes\\
    \tag*{$\displaystyle ...\otimes \left(\sum_{q_k}^{2^{m_k+n_k-1}}\lambda_{q_k,k}\prod_i\delta_{x,x_{k,i,q_k}}\ket{y_{k,1,q_1},...,y_{k,m_1,q_k}}\right)$}\\
    &=e^{i\alpha x}\left(\sum_{q_1}^{2^{m_1}}\lambda'_{q_1,1}\ket{y_{1,1,q_1},...,y_{1,m_1,q_1}}\right)\otimes...\otimes \left(\sum_{q_k}^{2^{m_k}}\lambda'_{q_k,k}\ket{y_{k,1,q_1},...,y_{k,m_1,q_k}}\right)\\
    &= \sum_{q=1}^{2^m}\lambda'_q\ket{y_{1,q},...,y_{m,q}}
\end{align*}

where the third line is obtained by induction hypothesis, and all $\lambda'$ match the ones obtained from the rewrite of token states.

\end{proof}

\begin{proof}[Proof of Lemma~\ref{lem:connected-universality}]
There exist several methods to build a diagram $D_f$ such that $\interp{D_f}=f$, using the universality of quantum circuits together with the map/state duality \cite{zxorigin}, or using normal forms \cite{jeandel2019generic}. The novelty here is that the diagram should be connected. This problem can be fairly simply dealt with:

Suppose we have such a $D_f$ that has several connected components. We can turn it into an equivalent diagram that is connected. Let us consider two disconnected components of $D_f$. Each of these disconnected components either has at least one wire, or is one of $\{
\begin{tikzpicture}
	\begin{pgfonlayer}{nodelayer}
		\node [style=gn] (0) at (0, 0) {};
	\end{pgfonlayer}
\end{tikzpicture}
\colorbox{pink}{missing file : g}}}
\tikz \node[style=glabel] {$\alpha$};,
\begin{tikzpicture}
	\begin{pgfonlayer}{nodelayer}
		\node [style=rn] (0) at (0, 0) {};
	\end{pgfonlayer}
\end{tikzpicture}
\colorbox{pink}{missing file : r}}}
\tikz \node[style=rlabel] {$\alpha$};\}$. In either case, we can use the rules of ZX ((I$_g$) or (H)) to force the existence of a green node. These green nodes in each of the connected components can be ``joined'' together like this:
$$
\InputIfFileExists{./figures/connecting-nodes.tikz}{}{{\color{red}\colorbox{pink}{missing file : connecting-nodes}}}
$$
It is hence possible to connect every different connected components of a diagram in a way that preserves the semantics.
\end{proof}

\begin{proof}[Proof of Proposition~\ref{prop:n-tokens}]
Using Lemma \ref{lem:connected-universality}, there exists a connected ZX-diagram $D'$ with $\mathcal I(D')=[a']$ and such that $\interp{D'}\ket0 = \sum_{q=1}^{2^n}\lambda_q\ket{x_{1,q},...,x_{n,q}}$. Consider now a derivation from the token state $(a'\downarrow 0)$ in $D\circ D'$:
$$\begin{array}{c@{\quad}||@{\quad}p{0.8\columnwidth}}
    \raisebox{-1.5em}{
\InputIfFileExists{./figures/D-Dprime.tikz}{}{{\color{red}\colorbox{pink}{missing file : D-Dprime}}}
} &
        $(a'\downarrow 0)\rewrites^* \sum_{q=1}^{2^n}\lambda_q\prod_{i=1}^n(a_i\downarrow x_{i,q})$\newline
        and\newline
        $(a'\downarrow 0)\rewrites^* \sum_{q=1}^{2^m}\lambda_q'\prod_{i=1}^m(b_i\downarrow y_{i,q})$
\end{array}$$
The first run comes from Proposition \ref{prop:single-token-input} on $D'$ which is connected. The second run results from Proposition \ref{prop:single-token-input} on $D\circ D'$ which is also connected. The proposition also gives us that:
$$\interp{D}\circ\left(\sum_{q=1}^{2^n}\lambda_q\ket{x_{1,q},...,x_{n,q}}\right) = \interp{D}\circ\interp{D'}\circ\ket0 = \interp{D\circ D'}\circ\ket0 = \sum_{q=1}^{2^m}\lambda_q'\ket{y_{1,q},...,y_{m,q}}$$
Finally, by confluence in $D\circ D'$, we get $\sum_{q=1}^{2^n}\lambda_q\prod_{i=1}^n(a_i\downarrow x_{i,q}) \rewrites^*\sum_{q=1}^{2^m}\lambda_q'\prod_{i=1}^m(b_i\downarrow y_{i,q})$ in $D$.
\end{proof}

\begin{proof}[Proof of Theorem \ref{thm:arbitrary-wire-init}]
First, let us single out $e$ in the diagram $D=
\InputIfFileExists{./figures/D-decomp-wire.tikz}{}{{\color{red}\colorbox{pink}{missing file : D-decomp-wire}}}
$. We can build a second diagram by cutting $e$ in half and seeing each piece of wire as an input and an output:

\InputIfFileExists{./figures/D-decomp-wire-cut.tikz}{}{{\color{red}\colorbox{pink}{missing file : D-decomp-wire-cut}}}
. We can easily see that a rewriting of the token states $(e\downarrow 0)(e\uparrow 0)$ and $(e\downarrow 1)(e\uparrow 1)$ in $D$ correspond step by step to a rewriting of the token states $(e_0\downarrow 0)(e_1\uparrow 0)$ and $(e_0\downarrow 1)(e_1\uparrow 1)$ in $D'$. We can then focus on $D'$, whose interpretation is taken to be \[\interp{D'}=\sum_{q=1}^{2^{m+n+2}}\lambda'_q\ketbra{y'_{1,q},...,y'_{m+1,q}}{x'_{1,q},...,x'_{n+1,q}}\]
such that 
\[(id^{\otimes m}\otimes \bra0)\circ\interp{D'}\circ(id^{\otimes n}\otimes \ket0)+(id^{\otimes m}\otimes \bra1)\circ\interp{D'}\circ(id^{\otimes n}\otimes \ket1) = \interp{D}\]
from which we get:
\begin{align*}
    \interp{D}=&\sum_{q=1}^{2^{m+n+2}}\lambda'_q\delta_{0,y'_{m+1,q}}\delta_{0,x'_{n+1,q}}\ketbra{y'_{1,q},...,y'_{m,q}}{x'_{1,q},...,x'_{n,q}}\\
    &\quad+\sum_{q=1}^{2^{m+n+2}}\lambda'_q\delta_{1,y'_{m+1,q}}\delta_{1,x'_{n+1,q}}\ketbra{y'_{1,q},...,y'_{m,q}}{x'_{1,q},...,x'_{n,q}}
\end{align*}
We now have to consider two cases:
\begin{itemize}
    \item $D'$ is still connected: By Proposition \ref{prop:single-token-input}, for $x\in\{0,1\}$:
    \begin{align*}
        (e_0\downarrow x)(e_1\uparrow x) 
        &\rewrites^* \sum_{q=1}^{2^{m+n+2}}\lambda'_q \delta_{x,x'_{n+1,q}} \prod_i(a_i\uparrow x'_{i,q})\prod_i(b_i\downarrow y'_{i,q})(e_1\downarrow y'_{m+1,q})(e_1\uparrow x)\\
        &\rewrites \sum_{q=1}^{2^{m+n+2}}\lambda'_q\delta_{x,y'_{m+1,q}}\delta_{x,x'_{n+1,q}} \prod_i(a_i\uparrow x'_{i,q})\prod_i(b_i\downarrow y'_{i,q})
    \end{align*}
    We hence have 
    \begin{align*}
    &(e_0\downarrow 0)(e_1\uparrow 0)
    \rewrites^* t_0 =
    \sum_{q=1}^{2^{m+n+2}}\lambda'_q\delta_{0,y'_{m+1,q}}\delta_{0,x'_{n+1,q}} \prod_i(a_i\uparrow x'_{i,q})\prod_i(b_i\downarrow y'_{i,q})\\
    &(e_0\downarrow 1)(e_1\uparrow 1) \rewrites^* t_1 = \sum_{q=1}^{2^{m+n+2}}\lambda'_q\delta_{1,y'_{m+1,q}}\delta_{1,x'_{n+1,q}} \prod_i(a_i\uparrow x'_{i,q})\prod_i(b_i\downarrow y'_{i,q})
    \end{align*}
    so $t_0+t_1$ corresponds to the interpretation of $D$.
    \item $D'$ is now disconnected: Since $D$ was connected, the two connected components of $D$ were connected through $e$. Hence, $D'$ only has two connected components, one connected to $e_0$ and the other to $e_1$. By applying Proposition \ref{prop:single-token-input} to both connected components, we get the desired result.
\end{itemize}
\end{proof}

\section{Proof of Section~\ref{sec:mixed-processes}}

\begin{proof}[Proof of Theorem~\ref{thm:simulation-cpm}]
Diffusion rules are trivial. Beware in the case of the Ground, as the CPM will produce a cup, the $\cpm$ does not produce a new token when applying the Trace-Out rule, meanwhile the $\rewrites$ machine will do two rewriting rules to pass through the cup.
\end{proof}

\end{document}